\documentclass[12pt]{iopart}

\usepackage{color}
\usepackage{booktabs}
\usepackage{float}
\usepackage{amssymb}
\usepackage{amsthm}
\usepackage{graphicx}
\usepackage[space]{grffile}
\usepackage{tensor}
\usepackage{amsfonts}
\usepackage{bm}
\usepackage[utf8]{inputenc}
\usepackage{slashed}
\usepackage{hyperref}
\usepackage[toc,page]{appendix}	%inserts an appendix
\usepackage{lipsum} % Used for inserting dummy 'Lorem ipsum' text into the template
\usepackage{comment}
\usepackage[percent]{overpic}
\usepackage{subfig}
\newcommand{\ext}[1]{\mathbf{d}\ensuremath{\mathbf{#1}}}
\providecommand{\abs}[1]{\vert#1\vert}
\newtheorem{thm}{Theorem}[section]
\newtheorem{dfn}{Definition}
\def\beq{\begin{eqnarray}}
\def\eeq{\end{eqnarray}}
\newcommand\Tstrut{\rule{0pt}{2.6ex}}         % = `top' strut
\begin{document}

\title[]{Bianchi cosmologies with $p$-form gauge fields}

\author{Ben David Normann$^1$, Sigbj\o rn Hervik$^1$, Angelo Ricciardone$^1$, Mikjel Thorsrud$^2$}

%\affil{Faculty of Science and Technology, University of Stavanger, 4036, Stavanger, Norway}

\address{$^1$Faculty of Science and Technology, University of Stavanger, 4036, Stavanger, Norway\\
$^2$Faculty of Engineering, \O stfold University College, 1757 Halden, Norway}
\ead{ben.d.normann@uis.no, sigbjorn.hervik@uis.no, angelo.ricciardone@uis.no, mikjel.thorsrud@hiof.no}
\vspace{10pt}
\begin{indented}
\item[]December 2017
\end{indented}

\begin{abstract}
In this paper the dynamics of free gauge fields in Bianchi type I-VII$_{h}$ space-times is investigated.  The general equations for a matter sector consisting of a $p$-form field strength ($p\,\in\,\{1,3\}$), a cosmological constant ($4$-form) and perfect fluid in Bianchi type I-VII$_{h}$ space-times  are computed using the orthonormal frame method. The number of independent components of a $p$-form in all  Bianchi types I-IX are derived and, by means of the dynamical systems approach, the behaviour of such fields in Bianchi type I and V are studied. Both a local and a global analysis are performed and strong global results regarding the general behaviour are obtained. New self-similar cosmological solutions appear both in Bianchi type I and Bianchi type V, in particular, a one-parameter family of self-similar solutions,``Wonderland ($\lambda$)'' appears generally in type V and in  type I for $\lambda=0$. Depending on the value of the equation of state parameter other new stable solutions are also found  (``The Rope'' and ``The Edge'') containing a purely spatial field strength that rotates relative to the co-moving inertial tetrad. Using monotone functions, global results are given and the conditions under which exact solutions are (global) attractors are found.
\end{abstract}
\vspace{2pc}
\noindent{\it Keywords}: $p$-form gauge fields, anisotropic space-times, Bianchi models, inflation, dynamical system, orthonormal frame.
%\maketitle

\section{Introduction}
Cosmological observations suggest that the universe is homogeneous and isotropic on large scales and the $\Lambda$CDM model seems to describe such a scenario with high accuracy. This has been confirmed by the~\textit{Planck} satellite, which measured a level of deviation from isotropy quite compatible with zero~\cite{Ade:2013nlj,Ade:2015hxq,Saadeh:2016sak}. This observational evidence seems to be in accordance with a scalar-driven inflationary epoch in which a scalar field, the inflaton, drives a quasi-de Sitter exponential phase of expansion. \\
However, the evidence of some unexpected features, called ``anomalies'', in the Cosmic Microwave Background (CMB), previously observed in the WMAP data~\cite{Bennett:2010jb}, and partly also confirmed by~\textit{Planck} data~\cite{Ade:2013nlj,Ade:2015hxq}, seems to suggest a possible deviation from isotropy and/or homogeneity at some point in the evolution of the universe. There are various anomalies that have been observed and their nature, cosmological or systematic, is still under debate~\cite{Ade:2013nlj}. Among them, the mutual alignment of the lowest multipole moments, the hemispherical asymmetry (between the total power on the North and South ecliptic hemispheres), and the dipole modulation of the CMB signal on very large scales, have been confirmed with a quite high level of significance. An appealing explanation, from the cosmological point of view, is a possible violation of the isotropy during the evolution of the universe. Considering a scenario based on a homogeneous scalar field, the invariance under spatial rotations remains unbroken, so, to violate the isotropy, it is necessary to modify the matter content of the primordial universe by introducing new field(s). 
Motivated by these observations, on one side theoretical models that can generate and sustain an anisotropic phase of expansion have gained attention and, on the other side, a deeper analysis of anisotropic space-times becomes necessary. Usually vector fields are employed in theoretical models to support an anisotropic evolution~\cite{Ford:1989me, Ackerman:2007nb,Golovnev:2008cf,Kanno:2008gn, Maleknejad:2011jw,Himmetoglu:2008zp, EspositoFarese:2009aj, Hervik:2011xm,Thorsrud:2012mu, Barrow:2005qv}. Recently a triplet of spin zero fields with a spatially-dependent vacuum expectation value~\cite{Endlich:2012pz, Bartolo:2013msa, Bartolo:2014xfa} have been investigated.

A more general understanding of the primordial sources that can generate an anisotropic evolution can be obtained by introducing $p$-form gauge fields in anisotropic space-times (Bianchi models). The case of a $2$-form field strength has already largely been subject of investigation in the literature. The possible existence of an homogeneous intergalactic magnetic field has led to investigation of source-free electromagnetism in the Bianchi models, see for instance~\cite{LeBlanc:1997ye} and references therein. Also, more recently, a general study of source-free Maxwell equations in Bianchi class B space-times has been conducted~\cite{Yamamoto:2011xy,Barrow:2011ct}. 

In the present paper, therefore, starting from a general $p$-form action, emphasis is put on the so far unstudied cases in a cosmological context: the $1$-form and $3$-form in anisotropic space-times. As explicitly shown, there is an equivalence between these two cases at the field strength level. Collectively they are henceforth referred to as a $j$-form (id est: $j=1$ or $j=3$).  Because of this equivalence the present study can be viewed as a general study of a massless \emph{inhomogeneous} scalar gauge field with a homogeneous gradient, in Bianchi space-times of type I-VII$_h$.  In fact this is a previously largely unexplored branch in the vast literature on cosmological scalar fields \cite{Fadragas:2013ina}. Previous investigations are primarily dealing with the possibility of realizing shear-free anisotropic cosmological models \cite{Carneiro:2001fz,Koivisto:2010dr,thorsrud17}. Beyond that a scalar field without a potential may appear to be of very limited relevance in cosmological model building, and has to the best of our knowledge been ignored. In this work the general evolution of cosmological models containing a $j$-form is studied for the first time. 

In cosmological model building containing an isotropy violating matter sector (see references above) the attention in the past has primarily been on \emph{locally rotationally symmetric} (LRS) models. In such models the space-time has an additional Killing vector representing invariance under spatial rotation. In LRS models the isotropy violating matter field, typically some kind of ``vector'', needs to be aligned with the Killing field, by assumption. It as been believed, perhaps, that such models are stable with respect to non-LRS type homogeneous perturbations. After all, a vector is intrinsically rotationally symmetric and one may assume it cannot source a more general non-LRS evolution of shear. This ignores the possibility that the vector may rotate. In the analysis presented here, state space is kept fully general from the start, within the respective Bianchi models. A  $\gamma$-law perfect fluid is included in the matter sector together with the $j$-form. In the case of Bianchi type I both a LRS (``Wonderland'') and non-LRS (``the Rope'' and ``the Edge'') self-similar solutions are found. It is shown that in the physically relevant parameter region $6/5<\gamma \le 2$, there is a unique non-LRS solution for each value of $\gamma$ that is stable, in fact a global attractor. 
Note that this range includes the radiation case $\gamma = 4/3$, in which, remarkably, the deceleration parameter $q=1$ is identical to the corresponding flat FLRW solution.\setcounter{footnote}{0}\footnote{In fact, for $\gamma\in(2/3, 4/3]$ all anisotropic attractors has a deceleration parameter $q=-1+\frac{3}{2}\gamma$ identical to the corresponding flat FLRW solution, see table \ref{tab:FpB1}.} 

Moreover, the Rope and the Edge contain a purely spatial field strength that rotates relative to the comoving inertial tetrad. Thus a phenomenon of ``vector rotation'' is present and in fact stable. In invariant LRS subspaces it turns out that Wonderland is the global attractor for the same $\gamma$ range. This demonstrates how essential it is to keep the analysis fully general, in order to identify the ``true'' dynamically preferred solution. To the best of our knowledge, this phenomenon of ``vector rotation'' at the background level is a qualitatively new feature of cosmological dynamical systems. To connect this explicitly to observational cosmology is not the scope of this paper. However, note that observational signatures of LRS vector models are quite limited, typically imprints in observables are themselves axisymmetric, often limited to quadrupole type modulations of CMB \cite{Copi:2013jna, Kim:2013gka, Naruko:2014bxa, Ramazanov:2016gjl, Bartolo:2012sd}.

One of the first new results of the current work is to write down the equations of motion and field equations for a system composed by a $j$-form, a  $\gamma$-law perfect fluid and a 4-form (cosmological constant) in Bianchi space-times of type I-VII$_h$. 

In the presence of a positive cosmological constant, cosmic no-hair theorems are formulated for the case of a $j$-form field strength present in the equations. The theorems imply that in the presence of a cosmological constant, or a $\gamma$-law barotropic fluid with $0<\gamma<2/3$, the universe asymptotically approaches the de Sitter universe, or a power-law inflating FLRW universe (quasi-de Sitter). 

Thereafter attention is devoted to analysing the general dynamics in Bianchi type I and V.
The no-hair theorems determine the late-time behavior of both the Bianchi type I and type V model, with $\Omega_{\Lambda}=0$ and $\gamma<2/3$.
For $\gamma\in(2/3,2)$, the complex dynamics of the system is tackled by use of a dynamical systems approach. The strategy then becomes that of finding equilibrium points and assessing their stability. Such analysis naturally allows for investigating solutions both close to and far away from isotropic solutions. A strong advantage in such an analysis,  therefore, is that one does not need the assumption that the universe has always been close to a Friedman-Lemaitre-Robertson-Walker (FLRW) background, as it is observed to be today. In addition to this, the dynamical systems approach enables a whole machinery to be implemented. Indeed, numerous monotone functions are found and used to find the global behaviour of solutions. Thus, the generality of such behaviour becomes apparent for each of the models considered. 

The method is implemented by writing down the system of equations using the orthonormal frame formalism. This is done for several reasons: first of all, the resulting equations are first order (partial differential equations); secondly, the physical meaning of the variables becomes more transparent. Two different gauge choices have been employed in the analysis, showing that $F-gauge$ is particularly suitable for the Bianchi type V system, whereas the $\Sigma_{3}$-gauge suits the analysis of the Bianchi type I system.

In the case of stiff matter ($\gamma=2$) there exists a set of solutions, the {\it Jacobs' Sphere}.
In the limit where the energy density of the perfect fluid is zero $(\Omega_{\rm pf}=0)$ it reduces to a {\it Jacobs' Extended Disk} solution where $\gamma\in[0,2)$. Simulations and analytics show that the  {\it Jacobs' Extended Disk} solution acts as a past attractor, as in the perfect fluid case.  As for late times, the Bianchi I solutions containing the $j$-form has already been mentioned. New one-parameter families of solutions are found in Bianchi type V: ``Wonderland'' ($\mathcal{P}_{\rm W}(\lambda)$) (where $\lambda=0$ is the Bianchi I particular case), which is an attractor in the range $2/3\leq\gamma<2$, with $\lambda$ restricted to $0\,\leq\,\lambda\,\leq\,\lambda_{\rm sup}\equiv\frac{\sqrt{3}}{4}\sqrt{2-\gamma }$, and ``Plane Waves'' ($\mathcal{P}_{\rm P. W.}(\Sigma_{+})$), with $\Sigma_{+}\in (-1,0]$, which are stable for $\gamma> 2/3-4/3\,\Sigma_{+}$. Indeed, using monotone functions, corresponding anisotropic hair theorems are given for Bianchi types I and V when $\Omega_{\lambda}=0$ and $2/3<\gamma<2$. These results show that generally the space-time has anisotropic hairs and the $j$-form field is dynamically significant during the evolution.  \\

The structure of the paper is as follows: in Section \ref{sec:pdyn} the dynamics of  a $p$-form is discussed in general. In Section \ref{sec:pclass} the classification of possible cases of interest in cosmology is explored.  In Section \ref{sec:formandbianchi}, an introduction to a further study of the $j$-form in the Bianchi models is made. Choice of frame is discussed and expansion normalized complex variables are introduced. A table of the number of independent components of the $j$-form in each particular Bianchi type is also given. In Section \ref{sec:systeq} the general evolution equations in an orthonormal frame are explicitly written down as a set of ordinary differential equations. Section \ref{sec:nohair} extends cosmic no-hair theorems to the case where differential forms are embedded in the anisotropic space-time. In Section \ref{sec:eqpoints} a definition of Equilibrium Points is provided, alongside certain gauge definitions. In Sections \ref{sec:dyn1} and \ref{sec:dyn5} the dynamical system is specialized to Bianchi type I and Bianchi type V, respectively. Dynamics and stabilities are assessed and described in detail both analytically and numerically.

To avoid confusion with differing conventions, in the rest of the paper $0,1,2,3$ and $4$-forms are labeled with the letters $\phi,\mathcal{A},\mathcal{B}, \mathcal{C}$ and $\mathcal{D}$, respectively. The forms are referenced at field strength level, and not at the underlying gauge potential level.

%%%%%%%%%%%%%%%%%%%%
\section{$p$-form dynamics}
\label{sec:pdyn}
%%%%%%%%%%%%%%
\subsection{The general $p$-form action}
Since the spaces considered (Bianchi models) are anisotropic, it is natural to consider an anisotropic matter sourcing. A natural candidate for such a source is that stemming from the general $p$-form action \cite{thorsrud17}

\begin{equation}
\label{action}
S_{\rm f}=-\frac{1}{2}\int{\mathcal{P}}\wedge\star\mathbf{\mathcal{P}}\,,
\end{equation}
 where $\mathbf{\mathcal{P}}$ is a $p$ -form constructed by the exterior derivative of a ($p-1$)-form $\mathcal{K}$. That is, 
\begin{equation}
\label{pForm}\fl
\mathbf{\mathcal{P}}=\ext{\mathcal{K}}=\frac{1}{(p-1)!}\nabla_{\mu_1}\mathcal{K}_{\mu_2\cdots\mu_{p}}\mathbf{\omega}^{\mu_1}\wedge\cdots\wedge\mathbf{\omega}^{\mu_{p}}=\frac{1}{p!}\mathcal{P}_{\mu_1\cdots\mu_{p}}\mathbf{\omega}^{\mu_1}\wedge\cdots\wedge\mathbf{\omega}^{\mu_{p}}\,.
\end{equation}
The Hodge dual is given by
\begin{equation}
\label{pFHodge}\fl
\star\mathbf{\mathcal{P}}=\frac{1}{p!(n-p)!}\eta_{\mu_1\cdots\mu_p\nu_1\cdots\nu_{n-p}}\mathcal{P}^{\mu_1\cdots\mu_p}\mathbf{\omega}^{\nu_1}\wedge\cdots\wedge\mathbf{\omega}^{\nu_{n-p}}=\frac{1}{(n-p)!}\ast \mathcal{P}_{\mu_1\cdots\mu_{n-p}}\mathbf{\omega}^{\mu_1}\wedge\cdots\wedge\mathbf{\omega}^{\mu_{n-p}},
\end{equation}
where $n$ is the dimension of the space and $\eta_{\mu_1\cdots\mu_n}=\sqrt{-g}\,\varepsilon_{\mu_1\cdots\mu_n}$. Here $g$ is the metric determinant and $\varepsilon_{\mu_1\cdots\mu_n}$ the standard anti-symmetric symbol of rank $n$. From \eref{pForm} and \eref{pFHodge} the explicit expressions for the $\mathbf{\mathcal{P}}$ components and the $\star\mathbf{\mathcal{P}}$ components are 
\begin{equation}
\label{DComp}\fl
\mathcal{P}_{\mu_1\cdots\mu_p}=p\nabla_{[\mu_1}K_{\mu_2\cdots\mu_{p}]}\quad\quad\textrm{and}\quad\quad \ast\mathcal{P}_{\nu_1\cdots\nu_{n-p}}=\frac{1}{p!}\eta_{\mu_1\cdots\mu_p\nu_1\cdots\nu_{n-p}}\mathcal{P}^{\mu_1\cdots\mu_p}.
\end{equation}
From this one obtains the action on component form
\begin{equation}
\label{Scomp}
S=-\frac{1}{2p!}\int\sqrt{-g}\,{\rm d}^4x\,\mathcal{P}^{\mu_1\cdots\mu_p}\mathcal{P}_{\mu_1\cdots\mu_p}.
\end{equation}
The energy-momentum tensor is
\begin{equation}\fl
\label{EnMomTens}
T_{\alpha\beta}\,\equiv\,-\frac{2}{\sqrt{-g}}\frac{\delta\left(\sqrt{-g}\mathcal{L}\right)}{\delta g^{\alpha\beta}}=\frac{1}{p!}\left[p\,\tensor{\mathcal{P}}{_\alpha^{\mu_2\cdots\mu_p}}\tensor{\mathcal{P}}{_{\beta\mu_2\cdots\mu_p}}-\frac{1}{2}g_{\alpha\beta}\tensor{\mathcal{P}}{^{\mu_1\cdots\mu_p}}\tensor{\mathcal{P}}{_{\mu_1\cdots\mu_p}}\right],
\end{equation}
 where $\mathcal{L}=-\frac{1}{2p!}\mathcal{P}^{\mu_1\cdots\mu_p}\mathcal{P}_{\mu_1\cdots\mu_p}$ is the Lagrangian density.
Now, equations of motion can be obtained by varying $\mathcal{L}$ with respect to the gauge field $\mathcal{K}$. One finds
\begin{equation}
\label{GenEoM}
\nabla_{\alpha_1}\mathcal{P}^{\alpha_1\cdots\alpha_p}=0\,.
\end{equation}
Energy conservation (contracted Bianchi identity) is now given by
\begin{equation}
\label{GenBianc}
\tensor{T}{^{\mu\nu}_{;\nu}}=0\,.
\end{equation}

\subsection{Exterior calculus}
The equations of motion \eref{GenEoM} and the Bianchi identity \eref{GenBianc}, both obtained from the action \eref{action}, may also be given in the language of exterior calculus; namely, 
\begin{equation}\fl
\label{dP}
\ext{\mathcal{P}}=0\quad\quad\rightarrow\quad\quad\nabla_{[\alpha_0}\mathcal{P}_{\alpha_1\cdots\alpha_p]}=0\quad\quad\textrm{Bianchi Identity.}
\end{equation}
Furthermore it is assumed that there are no sources, so the Hodge dual $\mathbf{\star \mathcal{P}}$ must be closed as well~\cite{gron07}. That is
\begin{equation}\fl
\label{StdP}
\ext{\star\mathcal{P}}=0\quad\quad\rightarrow\quad\quad\nabla_{\alpha_1}\mathcal{P}^{\alpha_1\cdots\alpha_p}=0\quad\quad\textrm{Equations of motion}.
\end{equation}

\paragraph{General properties of the $p$-form action: } Note that the theories derived from the general $p$-form action \eref{action}
respect the following properties: (i) gauge invariance
$\mathcal{L}\rightarrow\mathcal{L}$ under $\mathbf{\mathcal{K}}\rightarrow\mathbf{\mathcal{K}}+\ext{\mathbf{\mathcal{U}}}$, where $\mathcal{U}$ is a ($p-2$) -form;
(ii) only up to second order derivatives in equations of motion;
(iii) Lagrangian is up to second order in field strength $\mathbf{\mathcal{P}}$;
(iv) constructed by exterior derivatives of a $p$-form and 
(v) minimally coupled to gravity.

%\section{The three cases of space-time}
%%%%%%%%%%%%%%%%%%%%
\section{p-form classification}
\label{sec:pclass}
%%%%%%%%%%%%%%%%%%%%
From now on $\mathbf{\mathcal{P}}$ is required homogeneous: $\mathbf{\mathcal{P}}(t,\mathbf{x})\Rightarrow\mathbf{\mathcal{P}}(t)$. However, generally, the gauge field $\mathbf{\mathcal{K}}(t, \mathbf{x})$ is allowed to vary in space and time. This is different from \cite{Barrow:1996fh}, where the gauge potential is a function of time only. In order to classify the possible cases of $p$-form matter fields that can be constructed from the exterior derivative of a ($p-1$)-form, the following notation is introduced: $\{a,b\}$ where  $a$ denotes the rank of the $p$-form $\mathcal{P}$ and $b$ the rank of its Hodge dual $\star\mathcal{P}$. In four dimensional space-time ($a+b=4)$ there are three distinct cases to consider: (i)$\{2,2\}$, (ii) $\{3,1\}$ or $\{1,3\}$ and (iii) $\{4,0\}$. The degeneracy in (ii) is due to the symmetry of the equations  \eref{dP} and \eref{StdP}\setcounter{footnote}{0}\footnote{The reason why this degeneracy is not found in the case (iii) is because $\mathcal{P}\,\neq\,\ext{\mathcal{K}}$ in the case $\{0,4\}$, contrary to  \eref{pForm}. Thus one is left only with $\{4,0\}$. }. This symmetry can also be seen in the action \eref{action}, up to a prefactor.
\subsection{The $\{4,0\}$ case}
This case is equivalent to a model with a cosmological constant. According to \eref{DComp}, a $4-$form $\mathcal{D}$ can be constructed from a $3$-form. Defining $\star\mathcal{D}=c$, one finds
\begin{equation}\fl
\label{P4}
\mathcal{L}_{\rm 4f}=-\frac{1}{48}\mathcal{D}_{\mu_1\cdots\mu_4}\mathcal{D}^{\mu_1\cdots\mu_4}=\frac{1}{2}c^2\quad\rightarrow\quad T_{\mu\nu}^{\rm 4f}=\frac{1}{2}g_{\mu\nu}c^2.
\end{equation}
Also,
\begin{eqnarray}\fl
\label{dP4}
\ext{\mathcal{D}}&=0\quad\rightarrow\quad\textrm{identically satisfied}\\\fl
\label{dStP4}\ext{\star\mathcal{D}}&=0\quad\rightarrow\quad\nabla_{\mu}\,c=0\rightarrow\partial_\mu\, c=0.
\end{eqnarray}

\subsection{The $\{1,3\}$ and $\{3,1\}$ cases}
\paragraph{Case $\{1,3\}$:} According to \eref{DComp} a $1-$form $\mathcal{A}$ may be constructed from a $0$-form $\phi(t,\mathbf{x})$. One finds
\begin{equation}
\label{1form}
\mathcal{L}_{\rm 1f}=-\frac{1}{2}\mathcal{A}_{\mu}\mathcal{A}^{\mu}\quad\rightarrow\quad T_{\mu\nu}^{\rm 1f}=\mathcal{A}_{\mu}\mathcal{A}_{\nu}-\frac{1}{2}g_{\mu\nu}\mathcal{A}_\gamma\mathcal{A}^\gamma,
\end{equation}
and equations \eref{dP} and \eref{StdP} become
\begin{eqnarray}
\label{dD1}
\ext{\mathcal{A}}&=0\quad\rightarrow\quad\nabla_{[\mu}\mathcal{A}_{\nu]}=0,\\
\label{dStD1}\ext{\star\mathcal{A}}&=0\quad\rightarrow\quad\nabla_{\mu}\mathcal{A}^{\mu}=0.
\end{eqnarray}
These are the equations for a massless scalar field.
\paragraph{Case $\{3,1\}$: } According to \eref{DComp} a $3-$form $\mathcal{C}$ may be constructed from a $2$-form $\mathcal{B}$. Rewriting in terms of the Hodge dual components $\ast\mathcal{C}_\mu$ one finds 
\begin{equation}
\mathcal{L}_{\rm 3f}=-\frac{1}{12}\mathcal{C}_{\mu\nu\gamma}\mathcal{C}^{\mu\nu\gamma}\quad\rightarrow\quad T_{\mu\nu}^{\rm 3f}=\ast\mathcal{ C}_{\mu}\ast\mathcal{C}_{\nu}-\frac{1}{2}g_{\mu\nu}\ast\mathcal{C}_\gamma\ast\mathcal{C}^\gamma.
\end{equation} 
Equations \eref{dP} and \eref{StdP} become 
\begin{eqnarray}
\label{dD3}
\ext{\mathcal{C}}&=0\quad\rightarrow\quad\nabla_{\mu}\ast\mathcal{C}^{\mu}=0,\\
\label{dStD3}\ext{\star\mathcal{C}}&=0\quad\rightarrow\quad\nabla_{[\mu}\ast\mathcal{C}_{\nu]}=0.
\end{eqnarray}
Note the equivalence between the $\{1,3\}$ and $\{3,1\}$ cases. These are  the two cases of main interest in the analysis of the present paper.

\subsection{The $\{2,2\}$ case} 
The final option is a $2 -$form $\mathcal{B}$ constructed from a $1$-form $\mathcal{A}$, according to \eref{DComp}. This gives 
\begin{equation}
\mathcal{L}_{\rm 2f}=-\frac{1}{4}\mathcal{B}_{\mu\nu}\mathcal{B}^{\mu\nu}\quad\rightarrow\quad T_{\mu\nu}^{\rm 2f}=-\mathcal{B}_{\mu\gamma}\tensor{\mathcal{B}}{^\gamma_{\nu}}-\frac{1}{4}g_{\mu\nu}\mathcal{B}_{\gamma\delta}\mathcal{B}^{\gamma\delta},
\end{equation}
which coincides with the Lagrangian for source-free electromagnetism. Equations \eref{dP} and \eref{StdP} now give
\begin{eqnarray}
\label{dD2}
\ext{\mathcal{B}}&=0\quad\rightarrow\quad 3\nabla_{[\mu}\mathcal{B}_{\nu\lambda]}=0\,,\\
\label{dStD2}\ext{\star\mathcal{B}}&=0\quad\rightarrow\quad\nabla_{\mu}\mathcal{B}^{\mu\nu}=0\,,
\end{eqnarray}
These are the well known Maxwell's equations.

\subsection{Number of independent components in the Bianchi classification.}
In the following analysis the cases $\{1,3\}$ and $\{3,1\}$ will be taken into account. In order to include both scenarios, notation shall here, and throughout the rest of the paper, be such that $\mathcal{J}$ denotes either (i) the Hodge dual of a 3-form field strength $\mathcal{C}(t)=\ext{\mathcal{B}}$ (where $\mathcal{B}(t,\mathbf{x})$ is a $2$-form) or (ii) the $1$-form field strength $\mathcal{A}(t)=\ext{\phi}$ (where $\phi(t,\mathbf{x})$ is a scalar field). That both cases give rise to the same equations is evident from the previous section. The ``form-fluid'' will be referred to as the $j$-form fluid, where $j\,\in\{1,3\}$. 

Working out the equations of motion and the Bianchi Identity in either of the $\{1,3\}$ or $\{3,1\}$ cases, one may find the number of independent components of the field strength $\mathcal{J}$ allowed for in the different Bianchi space-times. Note that in this section no theory of gravity is yet assumed. The results are based solely on the equations \eref{dP} and \eref{StdP} for the $\{1,3\}$ or $\{3,1\}$ case. In particular note the equations which can be written out as
%\begin{equation}
%\label{IndepCompBianchiId}
%\partial_0\mathcal{J}_0=-3H\,\mathcal{J}_0-2a\,\mathcal{J}_1,
%\end{equation}
\begin{eqnarray}
\label{extStJab3}
&\partial_0\mathcal{J}_0=-3H\,\mathcal{J}_0-2a\,\mathcal{J}_1, \label{IndepCompBianchiId}\\
&\mathcal{J}_3 n_3+\mathcal{J}_2 a=0,\\
&\mathcal{J}_1 n_1=0,\\
&\mathcal{J}_2 n_2-\mathcal{J}_3 a=0.
\end{eqnarray}
The results are reported in table \ref{tab:J}, where the Bianchi classification is given in terms of eigenvalues $n_i$ of the matrix $n_{ab}$, introduced in the Behr decomposition (to be properly introduced later~\eref{strcoeff}). Here, $+$ and $-$ indicate positive and negative eigenvalues, respectively. The rightmost column gives the number of independent components. The number in parenthesis corresponds to the particular case $\mathcal{J}_0=0$, i.e a purely spatial field strength 1-form. Note that in Class B models ($a\neq0$) the number of independent components are actually reduced by two when $\mathcal{J}_0=0$, as a consequence of \eref{IndepCompBianchiId}.

It is essential to bear in mind that a certain theory of gravity might further restrict the number of allowed components. Bianchi type I will provide an example of this when employed with the Einstein Field Equations.

\begin{table}[t]
\centering
\begin{tabular}{lcccccccccc}
\toprule
\multicolumn{11}{c}{\textbf{Allowed field strength components in Bianchi space-times}} \\
\hline
Class&Type&$n_1$&$n_2$&$n_3$&$a$&$\mathcal{J}_0$&$\mathcal{J}_1$&$\mathcal{J}_2$&$\mathcal{J}_3$&\#\\
\hline
A &I&0&0&0&0&y&y&y&y&4(3)\\
&II&0&0&+&0&y&y&y&$0$&3(2)\\
&VI$_0$&0&+&--&0&y&y&$0$&$0$&2(1)\\
&VII$_0$&0&+&+&0&y&y&$0$&$0$&2(1)\\
&VIII&+&+&--&0&y&$0$&$0$&$0$&1(0)\\
&IX&+&+&+&0&y&$0$&$0$&$0$&1(0)\\
\toprule
B&III$\,$(VI$_{-1})$&0&+&--&+&y&y&y&y&3(1)\\
&IV&0&0&+&+&y&y&$0$&$0$&2(0)\\
&V&0&0&0&+&y&y&0&0&2(0)\\
&VI$_h\,(-1\neq\,h<0)$&0&+&--&$\sqrt{-h}$&y&y&$0$&$0$&2(0)\\
&VII$_h (h>0)$&0&+&+&$\sqrt{h}$&y&y&0&0&2(0)\\
\toprule
\end{tabular}
\caption{Components of $\mathcal{J}$ in the Bianchi models; $y$ denotes an allowed component. %The superscript $^*$ refers to a restriction on the allowance in terms of the equation \eref{IndepCompBianchiId} with $a=0$.The superscript $^{**}$ refers to a restriction in terms of the same equation, but this time with $a$ present. 
The last column gives the number of independent components in class A and class B, where the number in parenthesis corresponds to the particular case $\mathcal{J}_0=0$.}
\label{tab:J}
\end{table}

%%%%%%%%%%%%%%%%%%%%%%%%%%%%%%
\section{Sourcing anisotropy with 1-forms and 3-forms in General Relativity}
\label{sec:formandbianchi}
%%%%%%%%%%%%%%%%%%%%%%%%%%%%%%
As a theory of gravity, General Relativity is henceforth adapted. Thus the evolution is governed by the Einstein Field Equations. In particular
\begin{equation}
\label{Einst}
R_{\mu\nu}-\frac{1}{2}R\,g_{\mu\nu}=T^{\rm pf}_{\mu\nu}+T^{\rm f}_{\mu\nu}+T^{\rm 4f}_{\mu\nu}\,.
\end{equation}
Here $R_{\mu\nu}$ is the Ricci tensor components, $R=\tensor{R}{^\mu_\mu}$ is the Ricci scalar and $T^{\rm pf}_{\mu\nu}$ and $T^{\rm f}_{\mu\nu}$ the perfect fluid and $j$-form fluid\setcounter{footnote}{0}\footnote{Refer to the previous section for the meaning of $j$-form.} energy-momentum tensor components, respectively. The constant $8\pi G$ and $c$ are fixed to 1. A $4$-form is also added, playing the role of a cosmological constant (cf. \eref{dStP4}). From now on the $4$ -form will therefore be referred to as a cosmological constant. In a standard irreducible decomposition the notation used in this paper is such that
\begin{eqnarray}
&T^{\rm pf}_{\mu\nu}=\rho_{\rm pf} \,u_\mu u_\nu+p_{\rm pf}\,h_{\mu\nu}\,,\label{T1}\\
&T^{\rm f}_{\mu\nu}=\rho_{\rm f}\, u_\mu u_\nu+p_{\rm f} \,h_{\mu\nu}+2q_{(\mu}\,u_{\nu )}+\pi_{\mu\nu}\,,
\label{T2}\\
&T^{\rm 4f}_{\mu\nu}=\Lambda g_{\mu\nu}\,,
\end{eqnarray}
where $\rho_{\rm x}$ and $p_{\rm x}$ is the energy density and pressure of fluid $\rm x$ ($\rm x= perfect\,\, fluid, j-form$), respectively. Furthermore $h_{\mu\nu}=g_{\mu\nu}+u_{\mu}u_{\nu}$ gives the components of the projection tensor, $q_\mu$ is the heat flux components, $\pi_{\mu\nu}$ the anisotropic stress tensor components and  $u_\mu$ gives the fundamental observer's 4-velocity components. In general the form fluid will  be tilted\setcounter{footnote}{0}\footnote{A proper definition of a tilted fluid is a fluid for which the energy flux is non-vanishing. This will be the case here, and thus the $j$-form fluid is a tilted fluid in general.}. Also, since co-moving with the perfect fluid one has $\omega_{\mu\nu}=0$ (irrotational ) and $\dot{u}=0$ (no acceleration). \\
From \eref{EnMomTens} one finds that the energy-momentum tensor of the $j-$form is given by
\begin{equation}
\label{T3}
T^{\rm f}_{\mu\nu}=\mathcal{J}_\mu \mathcal{J}_\nu-\frac{1}{2}g_{\mu\nu}\,\mathcal{J}_{\alpha}\mathcal{J}^\alpha\,.
\end{equation}
The field strength $\mathcal{J}_{\alpha}$ will be decomposed according to 
\begin{equation}
\label{astJ}
\mathcal{J}_\alpha=-w\, u_\alpha+v_\alpha\,,
\end{equation}
where the 4-velocity $u_\alpha$ is time-like ($u_\alpha u^\alpha\,<\,0$), whereas $v_\alpha$ is defined to be orthogonal to $u_\alpha$ and therefore space-like ($v_\alpha v^\alpha\,>\,0$).\\
For the perfect fluid a barotropic equation of state,
\begin{equation}
\label{eospf}
p_{\rm pf}=(\gamma-1)\rho_{\rm pf}\phantom{000},\phantom{000}\gamma\,\in[0,2]\,,
\end{equation}
is assumed. For the $j$-form however, one finds the relation
\begin{equation}
\label{eos3f}
\fl p_{\rm f}=(\xi-1)\rho_{\rm f}\phantom{000}\textrm{where}\phantom{000}\xi=\frac{w^2-v^2/3}{w^2+v^2}+1\phantom{000}\rightarrow\phantom{000} \frac{2}{3}\,\leq\,\xi \,\leq\, 2\,.
\end{equation}
The range of $\gamma$ follows directly from requiring that $\mathcal{J}_\alpha\,\in\,\mathbb{R}$. Note that \eref{eos3f} is a dynamical equation of state, since the components of $\mathcal{J}$ in general change with time. The lower bound ($\xi \,=2/3$) is found for $w=0$ and the upper bound ($\xi \,=2$) is found for $v=0$. Note also that $w=v$ gives $\xi \,=\,4/3$, as in the case of electromagnetic radiation. 

Since, for simplicity, it is assumed that the three fluids do not interact, the three conservation equations
\begin{eqnarray}
&\nabla_{\mu} T_{\rm pf}^{\mu\nu}=0\label{Cons1}\,,\\
&\nabla_{\mu} T_{\rm f}^{\mu\nu}=0\label{Cons2}\,,\\
&\nabla_{\mu} T_{\rm 4f}^{\mu\nu}=0\label{Cons3}\,,
\end{eqnarray}
must be satisfied. The first equation will be calculated explicitly from \eref{T1} and the two last only implicitly through the corresponding Bianchi Identity.

\subsection{Bianchi models and choice of frame}
In dimension three there are nine different (classes of) Lie algebras -- these are the nine different Bianchi types I-IX. An intimate relationship between Killing vectors (the symmetries of the space) and Lie algebras may be established. In four-dimensional space-times, the Bianchi models, in addition to the Kantowski-Sachs model,  provide a nice way of classifying all anisotropic, yet spatially homogeneous universe models. In technical terms one says that the Bianchi models admit a three dimensional isometry group $G_3$ acting simply transitively on spatial hypersurfaces. The Kantowski-Sachs model is the only spatially homogeneoues model not allowing for a three dimensional group acting simply transitive on the spatial hypersurfaces. The line element of the Bianchi models can be written as
\begin{equation}
{\rm ds}^2=-{\rm d}t^2+\delta_{ab}\,\omega^a\omega^b\phantom{000}\textrm{where}\phantom{000}\ext{\omega}^a=-\frac{1}{2}\tensor{\gamma}{^a_{bc}}\omega^{b}\wedge\omega^c-\gamma^a_{~0c}{\rm d}t\wedge\omega^c.
\end{equation}
 $\{\omega^a\}$ is here a triad of 1-forms, and $\tensor{\gamma}{^a_{bc}}$ are the spatial structure coefficients of the Lie algebra characterizing the corresponding Bianchi type. The structure constants, which depend only on time, are typically split using the Behr decomposition in a vector $a_b$ and in a symmetric matrix $n^{ab}$
\begin{eqnarray}
\label{strcoeff}
\tensor{\gamma}{^m_{ab}}&=\tensor{\varepsilon}{_{abn}}n^{nm}+a_a\tensor{\delta}{_b^m}-a_b\tensor{\delta}{_a^m}\,.
\end{eqnarray}
By definition, elements of a Lie algebra satisfies the Jacobi Identity, 
\begin{equation}
\label{JacId}
\left[\mathbf{X},\left[\mathbf{Y},\mathbf{Z}\right]\right]+\left[\mathbf{Y},\left[\mathbf{Z},\mathbf{X}\right]\right]+\left[\mathbf{Z},\left[\mathbf{X},\mathbf{Y}\right]\right]=0\,,
\end{equation}
implying that vector $a_b$ lies in the kernel of the symmetric matrix $n^{ab}$
\begin{equation}
n^{ab}a_{b}=0\,.
\end{equation}
Refer to \cite{gron07}, Chap. 15 for details. The tetrad $\{\omega^\alpha\}$ is dual to the vector basis $\{\mathbf{e}_\alpha\}$, which must satisfy the relation
\begin{equation}
\label{bianchibasis}
[\mathbf{e}_\mu,\mathbf{e}_\nu]=\tensor{\gamma}{^\rho_{\mu\nu}}\mathbf{e}_\rho.
\end{equation}
The time direction is chosen orthogonal to the orbits of the isometry subgroup (i.e.: orthogonal to the three-dimensional hypersurfaces of homogeneity), and the fundamental observer's 4-velocity is aligned with this direction. It is given by 
\begin{equation}
\mathbf{u}=\frac{\partial}{\partial t}\,,
\end{equation}
where $t$ is the cosmological time.

A convenient frame in which to conduct the analysis is the orthonormal frame. As mentioned, such a frame will give first order evolution equations alongside a set of constraints which are useful to simplify the analysis. Without loss of generality, a choice is made such that $\e_1$ points in the direction of the vector $a_{b}$, leaving the remaining frame vectors $e_2$ and $e_3$ defined up to a rotation. This will become more transparent later, when the gauge freedom is discussed. Since a dynamical systems approach is adapted in this paper, the set of equations will be rewritten in expansion-normalized variables according to 
\begin{eqnarray}
\label{defENC}
&\fl\Sigma_{+}=\frac{\sigma_+}{H}\,,\quad\quad \Pi_{+}=\frac{\pi_+}{H^2}\,,\quad\quad \Omega_i=\frac{\rho_i}{3H^2}\phantom{0},\quad\quad A_i=\frac{a_i}{H}\phantom{000.},\nonumber\\
&\fl\Sigma_{-}=\frac{\sigma_-}{H}\,,\quad\quad\Pi_{-}=\frac{\pi_-}{H^2}\,,\quad\quad \Omega_\Lambda=\frac{\Lambda}{3H^2}\,,\quad\quad N_{+}=\frac{n_+}{H}\phantom{00.},\nonumber\\
&\fl\Sigma_{\times}=\frac{\sigma_\times}{H}\,,\quad\quad\Pi_{\times}=\frac{\pi_\times}{H^2}\,,\quad\quad V_i=\frac{v_i}{\sqrt{6}H}\,,\quad\quad N_{-}=\frac{n_-}{H}\phantom{00.},\\
&\fl\Sigma_2\,=\,\frac{\sigma_{2}}{H}\,,\quad\quad\Pi_2\,=\,\frac{\pi_2}{H^2}\,,\quad\quad\Theta=\frac{w}{\sqrt{6} H}\,,\quad\quad N_{\times}=\frac{n_\times}{H}\phantom{00.},\nonumber\\
&\fl\Sigma_3\,=\,\frac{\sigma_{3}}{H}\,,\quad\quad\Pi_3\,=\frac{\pi_3}{H^2}\,,\quad\quad\Xi_i=\frac{q_i}{3H^2}\,\phantom{0},\quad\quad\Sigma^2=\frac{\sigma_{ab}\sigma^{ab}}{6 H^2}.\nonumber
\end{eqnarray}
where $H$ is the Hubble parameter.
In this way the equations of motion become an autonomous system of differential equations and all equilibrium points will represent self-similar cosmologies (to be defined).

The above definitions slightly differ from the definitions used in \cite{Coley:2004jm, dynSys, relCos} and the precise decomposition of the different quantities has therefore explicitely been included in \ref{app:Decomp}.

The Bianchi space-times analyzed in the present paper (I-VII$_h$) admit an Abelian $G_2$ subgroup and this allows for a 1+1+2 split of the four dimensional space-time. As will become clear later, this translates into a $1+1+2$ decomposition of the Einstein Field Equations, the Jacobi and the Bianchi identities. When the orthonormal frame approach is applied to $G_{2}$ cosmologies, it is common to choose a \textit{group-invariant orbit-aligned} frame, i.e. an orthonormal frame which is invariant under the action of $G_{2}$~\cite{dynSys}. In this way the complete set of \textit{independent} basic variables reduces to 
\begin{equation}
\label{}
\{H,\sigma_{AB},\sigma_{1A},\Omega_1, n_{AB}, a\}\phantom{000}\textrm{and}\phantom{000}\{q_a,\pi_{AB},\pi_{1A},\rho_{\rm f},\rho_{\rm pf	}\}\,,
\end{equation}
where the capital letter indices $A, B$ run over 2 and 3, which are taken to be the two Killing vector fields chosen tangential to the group orbits of the $G_2$ subgroup. $\sigma_{11}$ and $\pi_{11}$ are derived from the trace-free property of these tensors, the isotropic pressures from equations of state and $a$ is chosen to be  equal to $(a,0,0)$. This also suggests the convention $A\,\equiv\,A_1=a_1/H$, which will be used henceforth. The energy densities, $\rho_{\rm pf}$ and  $\rho_{\rm f}$, refer to the perfect fluid and to the ``form-fluid'' respectively. This frame choice is further specified through rotations given by
\begin{equation}
\label{rot}
\Omega_{A}=\varepsilon_{AB}\sigma^{1B}\,.
\end{equation}
This equation follows from the propagation equation for $a$, which in turn comes from taking the trace of the Jacobi Identity \eref{JacId} applied to the vectors $(\mathbf{u},\textbf{e}_a,\textbf{e}_b)$ (\cite{gron07}, chap. 15.4).\setcounter{footnote}{0}\footnote{Note that for all Bianchi A models this relation becomes arbitrary, since $a$ vanishes. Thus in the type A models one may choose $\Omega_2$ and $\Omega_3$ differently. The choice implemented in the $G_2$ frame in this paper, however, will always be that of \eref{rot}.} By the above equation two of the frame rotations are specified. There remains in this way only one rotational gauge freedom (rotation of the frame around the $\mathbf{e}_1$-axis). This rotation is (when using  the angle $\phi$ which is constant on the orbit of $G_{2}$) given by the rotation
\begin{eqnarray}
\label{frameRot}
\eqalign{\mathbf{\tilde{e}}_2&=\phantom{0}\,\cos\phi\,\mathbf{e}_2+\sin\phi\,\mathbf{e}_3\,,\\
\mathbf{\tilde{e}}_3&=-\sin\phi\,\mathbf{e}_2+\cos\phi\,\mathbf{e}_3\,.}
\end{eqnarray}
Following~\cite{Coley:2004jm} the gauge freedom is left in the equations\setcounter{footnote}{0}\footnote{This differs from the general treatment in \cite{dynSys}, where gauge independent quantities are constructed.}
introducing the (expansion-normalized) local angular velocity $R_{a}$ of a Fermi-propagated axis with respect to the triad $\bf{e}_{a}$, with components
\begin{equation}
R_1\,\equiv\,\frac{\Omega_1}{H}=\phi'\phantom{000}\textrm{and}\phantom{000}\mathbf{R}_{\rm c}\,\equiv\,R_2+iR_3\,\,\equiv\,\frac{\Omega_2}{H}+i\frac{\Omega_3}{H}.
\end{equation}
Here $'$ denotes derivative with respect to dynamical time (to be properly defined in the next section). Since $\phi$ is given with respect to a frame for which $\phi'=0$, only $R_1$ is free in the $G_2$ aligned frame.\\
The complex variable  $\mathbf{R}_c$ is introduced in order to simplify the equations when the gauge symmetry is still not fixed. This is in accordance with~\cite{Coley:2004jm}\footnote{Note that there are some small conventional discrepancies in the current notation compared to that of \cite{Coley:2004jm}} and becomes a particularly useful tool in constructing gauge independent quantities. In particular
\begin{eqnarray}
\label{complexVar}
\eqalign{\mathbf{N}_\Delta=N_-+iN_\times\,,\quad\quad &\mathbf{\Phi_1}=\Xi_2+i\Xi_3\,,\quad\quad \mathbf{V}_c=V_2+iV_3\,,\\
\mathbf{\Sigma}_\Delta=\Sigma_-+i\Sigma_\times\,,\quad\quad &\mathbf{\Pi}_1=\Pi_2+i\Pi_3\,,\\
\mathbf{\Sigma}_1=\Sigma_2+i\Sigma_3\,,\quad\quad &\mathbf{\Pi}_\Delta=\Pi_-+i\Pi_\times.}
\end{eqnarray}
Some of the quantities introduced so far are independent under transformations over the remaining gauge freedom,\eref{frameRot}, whereas others change. To distinguish these quantities from each other, note the following two definitions.

{\dfn[Scalar]\label{def:Scalar} Any quantity invariant under the transformation \eref{frameRot} is said to be a scalar.}
{\dfn[Spin-n object] Any quantity $\mathbf{X}$ transforming such that $$\mathbf{X}\rightarrow\exp{(in\phi)}\mathbf{X}$$ under the transformation \eref{frameRot} is said to be a spin-n object.}

\paragraph*{}The above variables may now be classified as scalars or spin-n objects by looking at how they transform under  the gauge transformation (rotation) \eref{frameRot}:
\begin{eqnarray}
\label{VarTrans}
&\fl \{\Omega_{\Lambda}, A,N_+,\Sigma_+, \mathbf{\Sigma}_1, \mathbf{\Sigma}_\Delta,\mathbf{N}_\Delta\}\,&\rightarrow\,\{\Omega_{\Lambda}, A,N_+,\Sigma_+, \mathrm{e}^{i\phi}\mathbf{\Sigma}_1,\mathrm{e}^{2i\phi} \mathbf{\Sigma}_\Delta,\mathrm{e}^{2i\phi}\mathbf{N}_\Delta\}\,, \label{trans:1} \\
&\fl \{\Omega_{\rm pf}, \Omega_{\rm f},\Xi_1,\Pi_+,\mathbf{\Pi}_1,\mathbf{\Pi}_\Delta, \mathbf{\Phi_1}\}&\rightarrow\{\Omega_{\rm pf}, \Omega_{\rm f},\Xi_1,\Pi_+,\mathrm{e}^{2i\phi}\mathbf{\Pi}_\Delta,\mathrm{e}^{i\phi}\mathbf{\Pi}_1, \mathrm{e}^{i\phi}\mathbf{\Phi_1}\}\,,\\
&\fl \{\Theta,V_1,\mathbf{V}_c \}\,&\rightarrow\{\Theta,V_1,\mathrm{e}^{i\phi}\mathbf{V}_c \}.
\end{eqnarray}
Observe that the complex conjugates of the spin-n objects transform in a similar manner. In particular $\exp(ix)^*\,=\,\exp(-ix)$)\footnote{By such for instance $\mathbf{\Sigma}_\Delta\mathbf{\Sigma}_\Delta^*$ becomes a scalar quantity, since the exponentials cancel out.}. This makes it very easy to construct all sorts of physical variables (gauge independent quantities) from the spin-n objects. 

In the rest of the paper the variables used for the $j$-form matter content will be the four independent components $\{\Theta,V_1,\mathbf{V}_c,\mathbf{V}_c^*\}$, rather than the six energy-momentum tensor components built from them. However, even though they will not be used any further, it is instructive to give a list of the components in the standard irreducible decomposition:
\begin{eqnarray}
\label{enMomTenPartsC}
&\Omega_{\rm f}=\Theta^2+V_1^2+\abs{\mathbf{V}_c}^2\,,\nonumber\\
&\mathbf{\Pi}_\Delta=\sqrt{3}\mathbf{V}_c^2\,,\nonumber\\
&\mathbf{\Pi}_1=2\sqrt{3}V_1\mathbf{V}_c\,,\nonumber\\
&\Pi_+=\abs{\mathbf{V}_c}^2-2V_1^2\,,\nonumber\\
&\Xi_1=-2\,\Theta V_1\,,\nonumber\\
&\mathbf{\Phi}_1=-2\,\Theta\mathbf{V}_c.
\end{eqnarray}

%%%%%%%%%%%%%
\section{System of equations}
\label{sec:systeq}
%%%%%%%%%%%%%

The complete system of equations can be formulated through the variables~\eref{complexVar}.   The first set of equations is obtained through \eref{dP} and \eref{StdP} (with $p=1$ or $3$). The second set is composed by the shear propagation equations resulting from the Einstein Field Equations~\eref{Einst}. The third set is obtained from the contracted Bianchi identities \eref{Cons1} and \eref{Cons3} and the last set from the Jacobi Identity \eref{JacId}. The system reduces to 15 first order scalar ODEs (compactified below into 11 equations by the complex notation introduced above):
\begin{eqnarray}\fl
\textit{j}\textrm{-form eq.s }\eref{dP},\eref{StdP}\phantom{.}\cases{\label{FluidEqs}
V_1'=\left(q+2\Sigma_+\right)V_1-2\sqrt{3}\Re\{\mathbf{\Sigma}_1\mathbf{V}_c^*\}\,,\\
\mathbf{V}_c'=\left(q-\Sigma_+-iR_1\right)\mathbf{V}_c-\sqrt{3}\mathbf{\Sigma}_\Delta\mathbf{V}_c^*\,,\\
\Theta'=(q-2)\Theta-2AV_1\,,}\\\fl
\textrm{Einst. Eq.s }\eref{Einst}\quad\phantom{00}\cases{\label{EinstEq}
\mathbf{\Sigma}_1'=\left(q-2-3\Sigma_+-iR_1\right)\mathbf{\Sigma}_1-\sqrt{3}\mathbf{\Sigma}_\Delta\mathbf{\Sigma}_1^*+2\sqrt{3}V_1\mathbf{V}_c\,,\\
\mathbf{\Sigma}_\Delta'=(q-2-2iR_1)\mathbf{\Sigma}_\Delta+\sqrt{3}\mathbf{\Sigma}_1^2-2\mathbf{N}_\Delta\left(iA+N_+\right)+\sqrt{3}\mathbf{V}_c^2,\\
\Sigma_+'=\left(q-2\right)\Sigma_++3\abs{\mathbf{\Sigma}_1}^2-2\abs{\mathbf{N}_\Delta}^2+\abs{\mathbf{V}_c}^2-2V_1^2\,,}\\\fl
\textrm{En. cons. }\eref{Cons1},\eref{Cons3}\phantom{.}
\cases{\label{EnCons}
\Omega_\Lambda'=2(q+1)\Omega_\Lambda\,,\\
\Omega_{\rm pf}'=2\left(q+1-\frac{3}{2}\gamma \right)\Omega_{\rm pf}\,,}\\\fl
\textrm{Jacobi Id. }\eref{JacId2}\quad\quad\phantom{0}
\cases{\label{JacId2}
\mathbf{N}_\Delta'=\left(q+2\Sigma_+-2iR_1\right)\mathbf{N}_\Delta+2\mathbf{\Sigma}_\Delta N_+\,,\\
N_+'=\left(q+2\Sigma_+\right)N_++6\Re\{\mathbf{\Sigma}_\Delta^*\mathbf{N}_\Delta\}\,,\\
A'=\left(q+2\Sigma_+\right)A.}
\end{eqnarray}
where $\Re$ and $\Im$ represent the real and imaginary parts, respectively, and  where $'$ represents derivative with respect to dynamical time variable $\tau$ defined by
\begin{equation}
\qquad \frac{1}{H}=\frac{{\rm d} t}{\rm d\tau},
\end{equation}
where $t$ is the proper time.
Finally, the deceleration parameter $q$ has also been introduced in the above set of equations. It is defined through
\begin{equation}
\label{q1}
\qquad \dot{H}=-(1+q)H^2.
\end{equation}
These dynamical equations are subject to a set of 6 scalar constraints given by the four equations
\begin{eqnarray}\fl
& \sqrt{3}\mathbf{N}_\Delta^*\mathbf{V}_c-i\mathbf{V}_c^*\left(A+iN_+\right)=0\label{Constr1}\,,\\\fl
&1=\Sigma_{+}^2+\abs{\mathbf{\Sigma}_{\Delta}}^2+\abs{\mathbf{\Sigma}_1}^2+\Omega_{\rm pf}+\Theta^2+V_1^2+\abs{\mathbf{V}_c}^2+\Omega_{\Lambda}+A^2+\abs{\mathbf{N}_\Delta}^2\label{Constr2}\,,\\\fl
&2\,\Theta V_1=2\left(A\Sigma_+-\Im\{\mathbf{\Sigma}_\Delta\mathbf{N}_\Delta^*\}\right)\label{Constr3}\,,\\\fl
&2\,\Theta\mathbf{V}_c=\left(i\frac{N_+}{\sqrt{3}}-\sqrt{3}A\right)\mathbf{\Sigma}_1+i\mathbf{N}_\Delta\mathbf{\Sigma}_1^*\label{Constr4}\,,
\end{eqnarray}
as well as one group constraint determining the group-parameter $h$ in the type VI$_h$ and VII$_h$ models: 
\begin{equation}
A^2+h\left(3\abs{\mathbf{N}_\Delta}^2-N_+^2\right)=0. 
\end{equation}
In the above list of constraints, \eref{Constr1} comes from the Bianchi Identity and the others directly from the Einstein Field Equations. Also note that $q$ may be expressed as
\begin{equation}
\label{q}\fl
q=2\Sigma^2+(\frac{3}{2}\gamma -1)\Omega_{\rm pf}+2\Theta^2-\Omega_\Lambda,\quad\textrm{where}\quad\Sigma^2\equiv\Sigma_{+}^2+\abs{\mathbf{\Sigma}_{\Delta}}^2+\abs{\mathbf{\Sigma}_1}^2.
\end{equation}
This follows from \eref{q1} in combination with the Raychaudhuri's equation and the relation
\begin{equation}
2\Theta^2=\left(\frac{3}{2}\xi-1\right)\Omega_{\rm f}\,,
\end{equation}
where $\xi$ is defined in \eref{eos3f}.
This shows that the only component of $\mathcal{J}$ that enters in the equation for the deceleration parameter $q$ is the time component $\Theta$\\
%%%%%%%%%%%%%%%%%%%
\section{No-hair theorems for the $j$-form}
\label{sec:nohair}
%%%%%%%%%%%%%%%%%%%%
No-hair theorems that in previous literature has been established for the Bianchi space-times in the presence of a cosmological constant and a perfect fluid are in this section extended to the presence of the $j$-form in the equations. In particular it will be demonstrated that the cosmic no-hair theorem ~\cite{Wald:1983ky} is valid also in this case \setcounter{footnote}{0}\footnote{Note that an anisotropic fluid may sustain an inflationary phase of expansion if it violates the strong or dominant energy condition \cite{Maleknejad:2012as}. A $j$-form respects these energy conditions.}. To this end it is useful to formally define a de Sitter Universe.

\begin{dfn}[Flat de Sitter universe]
A flat de Sitter Universe is a universe which is maximally symmetric with flat spatial sections {\rm (} $\Sigma^2=A^2=\abs{\mathbf{N}_\Delta}^2=0${\rm )} and for which
\begin{equation}
\label{deSitter}
q=-1.
\end{equation} 
\end{dfn}
This result implies from \eref{q1} that $H'=0$. If $q=-1$ is imposed in \eref{q}, one may easily see by use of \eref{Constr2} that a de Sitter solution may be reached in the Bianchi types I-VII$_h$ if and only if $\Omega_\Lambda=1$ or if $\Omega_{\rm pf}=1\,,\,\gamma=0$.\\
Having this definition, the cosmic no-hair theorem can be extended to the case where $j$-form matter is contributing to the content of the universe
\begin{thm}[First no-hair theorem]
\label{thmDeSitter}
All Bianchi space-times I-VII$_h$ with a $j$-form, a non-phantom perfect fluid\setcounter{footnote}{0}\footnote{A perfect fluid is said to be phantom if $\gamma \,<\,0$. }  and a positive cosmological constant will be asymptotically de Sitter with $\Omega_\Lambda=1$ in the case where $\gamma\,>\,0$ (and similarly $\Omega_\Lambda+\Omega_{\rm pf}=1$ in the case where $\gamma=0$).
\end{thm}
\begin{proof}
Combining \eref{Constr2} and \eref{q} one finds 
\begin{equation}
\gamma \,\geq\,0\quad\Rightarrow\quad q+1\,\geq\,0.
\end{equation}
Thus, from \eref{EnCons} it is evident that $\Omega_{\Lambda}$ increases monotonically. If, therefore,  for some instance of time $\tau=\tau_0$ one has $\Omega_{\Lambda}\,>0$, then one must also have $\Omega_{\Lambda}\,>\,0\quad\forall\quad\tau\,>\,\tau_0$. Additionally, since $\Omega_{\Lambda}$ is monotonically increasing, and bounded by $\Omega_{\Lambda}\,\leq\,1$, \eref{Constr2}, one must have
\begin{equation}
\lim_{\tau\rightarrow\infty}\Omega_{\Lambda}'=0.
\end{equation}
By \eref{q} this gives $q=-1$. It then follows that
\begin{equation}
\vert\mathbf{N}_\Delta\vert^2+A^2+\frac{3}{2}\gamma \Omega_{\rm pf}+3\Theta^2+\vert\mathbf{V}_c\vert^2+\vert V_1\vert^2+\vert\mathbf{\Sigma}_\Delta\vert^2+\Sigma_1^2+\Sigma_+^2=0\,,
\end{equation} which in turn implies that all terms must vanish since all terms are strictly positive. For $\gamma\,>\,0$ it thus follows that $\lim_{\tau\rightarrow\,\infty}\Omega_\Lambda=1$ (first eq. in \eref{EnCons}). For $\gamma\,=\,0$, one finds that $\Omega_{\rm pf}$ is just another cosmological constant, and thus the results are dynamically the same, with $\Omega_{\Lambda}+\Omega_{\rm pf}=1$.
\end{proof} 
A similar but less general theorem holds also in the case of a perfect fluid with $0\leq\,\gamma\,<\,2/3$:
\begin{thm}[Second no-hair theorem]
\label{NoHair}
All Bianchi space-times I-VII$_h$ with a $j$-form, a non-phantom perfect fluid $\Omega_{\rm pf}$ with equation of state parameter $0\leq\,\gamma\,<\,2/3$ will be asymptotically quasi de Sitter with $q=\frac{3}{2}\gamma-1\,<\,0$.
\end{thm}
\begin{proof}
From equation \eref{q} with $\Omega_\Lambda=0$ it is evident that $q\,\geq\,(3\gamma-2)\Omega_{\rm pf}/2$. Using this in the evolution equation for $\Omega_{\rm pf}$ in \eref{EnCons}, and using the natural logarithm, one finds
\begin{equation}
(\ln\Omega_{\rm pf})'\,\geq\,3\left(\frac{2}{3}-\gamma\right)\left(1-\Omega_{\rm pf}\right)\,\geq\,0\quad\forall\quad0\leq \gamma<\frac{2}{3}.
\end{equation}
Thus $\lim_{\tau\rightarrow\infty} \Omega_{\rm pf}=1$, with $q=\frac{3}{2}\gamma-1\,<\,0$.
\end{proof}

%%%%%%%%%%%%%%
\section{Equilibrium points and choice of gauge}
\label{sec:eqpoints}
%%%%%%%%%%%%%

As anticipated, in the dynamical systems approach a relevant role is given to the equilibrium points, as they provide exact solutions of the system. In order to formally define an equilibrium point in a gauge independent manner,  consider the definition of a scalar given in Def. \ref{def:Scalar}. Then, a gauge independent definition of an equilibrium point is  

\begin{dfn}[Equilibrium point] 
An equilibrium point P is a set on which all scalars are constants on $P$ as functions of $\tau$. 
\end{dfn}
Equilibrium points found from \eref{FluidEqs}-\eref{JacId2} represent self-similar cosmological models. A definition of self-similar is as follows.
\begin{dfn}[Homothety and Self-similar space-time] 
A self-similar space-time is a space-time possessing a proper homothety. A vector field $\mathbf{H}$ is said to be a (proper) homothety if
\begin{equation}
\pounds_{\mathbf{H}}\mathbf{g}=k\,\mathbf{g},
\end{equation}
where $\mathbf{g}$ is the metric tensor, $k$ a (non-zero) constant and $\pounds$ denotes the Lie derivative.
\end{dfn}
The dynamical system is of the form 
\begin{equation}
X'=F(X), \qquad {\cal C}_i(X)=0,
\end{equation}
where $X$ is the $n$-dimensional state space vector of the system, ${\cal C}_i(X)=0$ is the set of constraints, and $F$ is an $n$-dimensional vector function. The local stability of the self-similar cosmological solutions represented by equilibrium points, $X_0$ (where $F(X_0)=0$), may now be computed by looking at displacements from such points to linear order:
\begin{equation}
(\delta X)'=J\,(\delta X).
\end{equation}
Here $J$ is the Jacobian matrix of the system. The eigenvalues $l$ are given by the equation
\begin{equation}
\det(J-I\,l)=0.
\end{equation}
This method will explicitly be applied to the Bianchi types I and V. Before that, however, another important point to clarify in the analysis is the choice of gauge. As seen in~\eref{FluidEqs}-\eref{JacId2}, the gauge freedom $R_1$ is left in the equations. Gauge freedom represents unphysical degrees of freedom. This freedom may be used in order to simplify the analysis. In the following the two gauges used in the analysis are discussed:

\paragraph{F-gauge ($R_1=0$):} This is in some sense a quite physical gauge\setcounter{footnote}{0}\footnote{Remember that $R_1\equiv\phi'$.}: $R_1$ specifies the angular velocity compared to a Fermi-Walker propagated frame, so equating this to zero means that one plane is following the frame of gyroscopes\footnote{To fully align with the gyroscope frame one must additionally have $\Omega_2=\Omega_3=0$. As it is, however, this will automatically be fulfilled in the type V case.}. Note that there is still a $U(1)$ gauge freedom left: the initial configuration of the frame around the $G_2$ orthonormal axis.

\paragraph{$\Sigma_3$-gauge ($R_1=\sqrt{3}\Sigma_\times$):}In this gauge $\mathbf{\Sigma}_1$ is imposed to be purely real, so $\Re \{\mathbf{\Sigma}_{1}\}=0$. This gauge choice becomes natural in the analysis of type I in a $G_2$ frame aligned such that $\mathbf{V}_c=0$.

\paragraph{} The following sections seek to explore the implementation of the theory with certain Bianchi types. Analyzing all Bianchi types I-VII$_h$ will be too extensive for this paper and hence only some of the simplest (and important) types will be analyzed here. Among the class B types (nonzero $a$), the simplest extension to the commonly assumed open FLRW background is the type V. It borders to Bianchi type I, which in turn is the simplest extension of the flat FRW model and belong to the class A types ($a=0$). Some of the equilibrium points important in the  type V analysis will therefore prove to be of type I. Hence both of these types will be analyzed in the following sections where quilibrium points are found and the stability analyzed. 

%%%%%%%%%%%%%%%%%

\section{Dynamical system in Bianchi type I}
\label{sec:dyn1}
%%%%%%%%%%%%%%%%%

The Bianchi type I model is characterized by flat spatial sections:

\begin{equation}
\label{B1}
A=\mathbf{N}_{\Delta}=N_+=0\,.
\label{BianchiIspec}
\end{equation}
In the following, the spatial components $\{V_1, \mathbf V_c=V_2+i V_3\}$ associated with the field strength will be referred to as the \emph{spatial part}, or sometimes loosely as the \emph{vector}, and $\Theta$ as the \emph{temporal part}. 

Below some peculiarities of Bianchi type I are examinded in subsection \ref{sec:decoupling}. First it is shown that state space decouples into two branches. Thereafter the choice of frame is chosen, which is different in the two branches. In the next two subsections attention will be given to each of the two branches and self-similar solutions that correspond to equilibrium points in the dynamical system will be derived. Thereafter, an analysis of state space will be conducted in subsection \ref{subs:B1_class}, before strong global results are derived in subsection \ref{sec:global:BI}.

\subsection{Decoupling and choice of frame \label{sec:decoupling}}
From \eref{Constr3} and \eref{Constr4} subject to \eref{BianchiIspec} it is evident that the temporal and spatial components of the field strength decouple in Bianchi type I:
\begin{itemize}
	\item \underline{Perfect branch:} The field strength $j$-form is purely temporal, i.e. $\Theta$ is generally nonzero and $V_1=\mathbf{V}_c=0$. This corresponds to a 1-form  field strength constructed from a homogeneous massless scalar field.
	\item \underline{Vector branch:}  The field strength $j$-form is purely spatial, i.e. $V_1$ and $\mathbf{V}_c$ are generally non-zero and $\Theta=0$.  
\end{itemize}
In fact, the condition for a vanishing gauge field energy flux is that its field strength is purely spatial ($\Theta=0$) or purely temporal ($V_1=\mathbf V_c=0$) \cite{thorsrud17}, as seen from the two last equations in \eref{enMomTenPartsC}. The decoupling can therefore be interpreted as a physical restriction on the gauge field: energy flux on spatial sections of homogeneity is not possible in Bianchi type I. 

Because of the decoupling one obtain two different dynamical  systems, that require a separate analysis. The natural choice of frame is different in these two cases. In the case of Bianchi type I where the group $G_3$ is Abelian, there is no unique subgroup $G_2$. In this situation there are two natural possibilities for choice of frame: 

\begin{itemize}	
	\item \underline{Diagonal shear frame:} 
	In this case the tetrad is aligned with the shear eigenvectors, so that the shear tensor takes a diagonal form. In order to consider a general state space, this choice requires three independent spatial components of the field strength, since it will in general not be aligned with the eigenvectors of the shear tensor. The general Bianchi type I equations for this frame is given in \ref{App:diag}. These equations lead to dynamical systems that are effectively 3-dimensional and 5-dimensional in the perfect branch and the vector branch, respectively. This frame will be adopted only for the perfect branch, in which case the dynamical system can also be derived directly from \eref{FluidEqs}-\eref{JacId2}. The resulting dynamical system is presented in subsection \ref{ch:puretemporal}.  	
	\item \underline{Vector aligned frame:} In the vector branch the frame can be aligned with the field strength ``vector''. This ``vector aligned frame'' will be the choice of frame for the vector branch.  Specifically, the basis vector $\mathbf e_1$ will be aligned with the $j$-form so that $\mathbf V_c=0$.
	Note that the general equations \eref{FluidEqs}-\eref{JacId2} preserve the initial choice $\mathbf V_c=0$. In this case, it should not come as a surprise that restricting the analysis to an inertial frame would imply restrictions on state space. As will be seen later; with generic initial conditions the vector will rotate, so in this sense frame rotation is required to preserve the alignment with the vector. The resulting dynamical system is effectively 5-dimensional and will be presented in subsection \ref{vectorbranch}. 	 
\end{itemize}

\subsection{The perfect branch ($V_1=\mathbf{V}_c=0$) \label{ch:puretemporal}}
With $V_1=\mathbf{V}_{c}=0$ the diagonal shear frame equations \eref{FinalEqs1}-\eref{frameDef2} reduce to\setcounter{footnote}{0}\footnote{This can also be derived directly from \eref{FluidEqs}-\eref{JacId2} by inserting perfect branch / Bianchi type I specifications ($V_1=\mathbf V_c = N_+ = \mathbf N_\Delta = A =0$) and choosing an initial orientation of the tetrad so that the shear is diagonal. The diagonality is preserved by choosing $R_1=0$.}
\begin{eqnarray}
\label{FinalEqs1Temp}
&\Sigma_+'=-\left(2-q\right)\Sigma_+,\\
&\Sigma_-'=-\left(2-q\right)\Sigma_-,\\
&\Theta'=-(2-q)\Theta,\\
\label{FinalEqs4Temp}
&\Omega_{\rm pf}'=2\left(q+1-\frac{3}{2}\gamma, \right)\Omega_{\rm pf},
\end{eqnarray}
subject to one constraint
\begin{equation}
\Sigma_+^2+\Sigma_-^2+\Theta^2+\Omega_{\rm pf}=1.
\label{BI:con1}
\end{equation}
For this system
\begin{equation}
q=2(\Sigma_+^2+\Sigma_-^2)+(\frac{3}{2}\gamma -1)\Omega_{\rm pf}+2\Theta^2.
\end{equation}
It is straight forward to solve the equations \eref{FinalEqs1Temp}-\eref{FinalEqs4Temp} exactly. For $\left(\Sigma_+,\Sigma_-,\Theta,\Omega_{\rm pf}\right)$ the solution is 
\begin{equation}\fl
\label{SolB1}
\left(\Sigma_+,\Sigma_-,\Theta,\Omega_{\rm pf}\right)=\left(\frac{\Sigma_{+,0}e^{-2\tau}}{\theta},\frac{\Sigma_{-,0}e^{-2\tau}}{\theta},\frac{\Theta_0e^{-2\tau}}{\theta},\frac{\Omega_{\rm pf,0}e^{-(3\gamma-2)\tau}}{\theta^2}\right), 
\end{equation}
where 
\begin{equation}
\theta= \left[(\Sigma_{+,0}^2+\Sigma_{-,0}^2+\Theta_0^2)e^{-4\tau}+\Omega_{\rm pf,0}e^{-(3\gamma-2)\tau}\right]^{\frac 12},
\end{equation}
and the initial values $(\Sigma_{+,0},\Sigma_{-,0},\Theta_0,\Omega_{\rm pf,0})$ satisfy the constraint $\Sigma_{+,0}^2+\Sigma_{-,0}^2+\Theta_0^2+\Omega_{\rm pf,0}=1$. Note that $q\in[-1,2]$. Also note that these variables represent observables relative to an inertial tetrad, since the frame specifications \eref{frameDef1}-\eref{frameDef2} allowed us to set the frame rotation to zero, i.e. $\Omega_a=0$ for $a\in\{1,2,3\}$. The exact solutions thus determine the evolution completely. 
They admit four sets of self-similar solutions: \emph{flat FLRW}, \emph{Jacobs' Sphere}, \emph{Jacobs' Extended Disk} and \emph{Kasner}:
\begin{itemize}
	\item For $-1\leq\,q\,<\,2$ all observables except $\Omega_{\rm pf}$ vanish asymptotically as  $\tau\,\rightarrow\,\infty$. Since $\Omega'$ is a monotonic increasing function, the Hamiltonian constraint gives $\lim_{\tau\rightarrow\infty}\Omega=1$. Therefore, the flat FLRW universe is reached asymptotically.
	\item For $q=2$ all derivatives vanish except $\Omega_{\rm pf}'$. Invoking the Hamiltonian constraint there are two options: 
	\begin{equation}\fl
		\label{d}
		\Omega_{\rm pf}\,\neq\,0:\quad 1=\Sigma_+^2+\Sigma_-^2+\Theta^2+\Omega_{\rm pf} \phantom{000}\textrm{with}\phantom{000}\gamma =2\rightarrow\quad Jacobs'\;\; Sphere\,,
	\end{equation}
	or
	\begin{equation}\fl
		\label{k}
		\Omega_{\rm pf}\,=\,0:\quad 1=\Sigma_+^2+\Sigma_-^2+\Theta^2 \phantom{000}\textrm{with}\quad\gamma \,\in\,[0,2)\rightarrow\quad Jacobs'\;\; Ext. \;\;Disk.\,
	\end{equation}
	Note that in Jacobs' Extended Disk $\gamma$ is defined on a half-open interval, so that there is no overlap with Jacobs' sphere.
	\item Setting $\Theta=0$ in Jacobs' Sphere gives (the ordinary) \emph{Jacobs' disk}. Setting $\Theta=0$ in Jacobs' Ext. Disk  gives \textit{Kasner} vacuum solution.  
\end{itemize}
  
In table \ref{tab:FpB1Th} the above equilibrium sets are summarized. The names Jacob's Sphere and Jacob's Extended Disk are so chosen attempting at reflecting the mere extension that these equilibrium points represent of the already known Jacob's Disk.

\begin{table}[H]
	\centering
	\resizebox{\textwidth}{!}{\begin{tabular}{llccccc}
			\toprule
			\multicolumn{7}{c}{\textbf{Equilibrium sets in Bianchi type I with pure temporal field ($V_1=\mathbf{V}_c=0$)}} \\
			\hline
			Name (abbr.) & $q$ &$\gamma $&$\Omega_{\rm pf}$&$\Sigma_+$&$\Sigma_-$&$\Theta$ \\
			\hline
			flat FLRW & $\frac{3}{2}\gamma -1$ & $[0,2)$&1&0&0&0 \Tstrut\\ 
			Kasner (K.) & $2$ & $[0,2)$&0&$[-1, 1]$&$\pm \sqrt{1-\Sigma_+^2}$&0 \\
			Jacobs' Ext. D. (J.E.D.) &  $2$&$[0,2)$&0&$[-1,1]$&$[-1,1]$& $\pm \sqrt{1-\Sigma_+^2 - \Sigma_-^2}$ \\
			Jacobs' Sphere (J.S.) & $2$ & $2$ & $[0,1]$ & $[-1,1]$ & $[-1,1]$ &  $\pm \sqrt{1-\Sigma_+^2 - \Sigma_-^2-\Omega_{\rm pf}}$ \\ \hline
	\end{tabular}}
	\caption{Summary of equilibrium sets in Bianchi type I for the perfect branch $V_1=\mathbf{V}_c=0$. }
	\label{tab:FpB1Th}
\end{table}

\subsection{The vector branch ($\Theta=0$) \label{vectorbranch}}
Instead of the diagonal shear frame, a ``vector aligned frame'' will be used in this case. As mentioned above, in Bianchi type I there is no unique subgroup $G_2$ of the Abelian $G_3$. By convenience a frame aligned with the field strength ``vector'', i.e. $\mathbf{V}=V^1\mathbf{e}_1$, is chosen. The equations imply that $\mathbf{V}_c$ remains zero for all later times. The equations are obtained by inserting the type I specifications \eref{B1} into \eref{FluidEqs}-\eref{JacId2} and $\eref{Constr1}$-$\eref{Constr4}$. The frame rotations $R_2$ and $R_3$ are now given by \eref{rot}.\setcounter{footnote}{0}\footnote{The Jacobi Identity does not impose \eref{rot} in the type I case, and this choice is thus a choice of gauge.} Explicitly these equations are written down in \ref{App:G2B1}.  

$R_1$ is the normalized frame rotation around the $\mathbf{e}_1$ axis. Specifying to the $\Sigma_3$-gauge,
\begin{equation}
\label{S3}
R_1=\sqrt{3}\Sigma_\times,
\end{equation}
reduces $\Sigma_2$ to a monotone function, see \eref{heihei} in \ref{App:G2B1}. Thus, by the above gauge choice, $\Sigma_2$ must now remain zero if it initially vanishes: $\Sigma_2(\tau_0)=0$. This initial condition may be met by appropriately fixing the remaining freedom in the gauge (\ref{S3}): the angle $\phi$ which represents the freedom in orientation of the tetrad at a given instant. Now all the gauge freedom has been used, and, as seen below, one is left with a five-dimensional dynamical system. %As a consistency check, this is the same dimensionality as in the diagonal shear frame (for the vector branch $\Theta=0$).

The resulting dynamical system expressed in six real variables is
\begin{eqnarray}
&\Sigma_+'=(q-2)\Sigma_++3\Sigma_3^2-2V_1^2, \label{B1_G2:first}\\
&\Sigma_-'=(q-2)\Sigma_-+\sqrt{3}(2\Sigma_\times^2-\Sigma_3^2),\\
&\Sigma_\times'=(q-2-2\sqrt{3}\Sigma_-)\Sigma_\times,\\
&\Sigma_3'=\left(q-2-3\Sigma_++\sqrt{3}\Sigma_-\right)\Sigma_3,\\
&V_1'=\left(q+2\Sigma_+\right)V_1,\\
&\Omega_{\rm pf}'=2\left(q+1-\frac{3}{2}\gamma \right)\Omega_{\rm pf}, \label{B1_G2:last}
\end{eqnarray}
subject to one constraint equation
\begin{eqnarray}
\label{HamB1V}
&1=\Omega_{\rm pf}+\Sigma_+^2+\Sigma_-^2+\Sigma_\times^2+\Sigma_3^2+V_1^2.
\end{eqnarray}
For this system
\begin{equation}
q=2(\Sigma_+^2+\Sigma_-^2+\Sigma_\times^2+\Sigma_3^2)+(\frac{3}{2}\gamma -1)\Omega_{\rm pf}.
\end{equation}

Thus the dynamical system in the vector branch has two more degrees of freedom than in the perfect branch. This reflects the fact that a rotating vector has three physical degrees of freedom, two more than the field strength in the perfect branch, where $\Theta$ is the only variable. Therefore it can be expected that the effective dynamical system with a \emph{non-rotating} vector is three-dimensional, like the perfect branch. This will be verified below, where ``vector rotation'' will be defined formally after paying closer attention to the gauge fixing.

%Note that since the vector is fixed to be aligned with the basis vector $\mathbf e_1$, the vector is rotating relative to a co-moving inertial observer iff the frame rotation along $\mathbf e_2$ and $\mathbf e_3$ is non-zero, i.e. $\mathbf{R}_{\rm c}\,\equiv\, R_2+iR_3 \neq 0$. 

The gauge fixing above, that resulted in the five-dimensional system (\ref{B1_G2:first})-(\ref{B1_G2:last}), assumes $\mathbf{\Sigma}_1=\Sigma_2+i\Sigma_3 \neq 0$. Note that the statement $\mathbf{\Sigma}_1 \neq 0$  is gauge independent according to (\ref{trans:1}) and preserved in time according to (\ref{heihei})-(\ref{heihei2}). In the case $\mathbf{\Sigma}_1=0$, there remains one gauge degree of freedom among the four remaining independent variables  in the system (\ref{B1_G2:first})-(\ref{B1_G2:last}). The redundant degree of freedom can be removed by choosing an initial orientation of the tetrad so that $\Sigma_\times = 0$. Thus the effective dynamical system is three-dimensional if $\Sigma_3=0$, as in the perfect branch.  In that case the shear is diagonal and hence the frame rotation is zero. One can conclude that the vector, which is aligned with the frame basis vector $\mathbf e_1$, is non-rotating if $\Sigma_3=0$. That is to say: the direction of the vector is stable as seen by a co-moving inertial observer. In the same way the vector is rotating if $\Sigma_3\neq 0$, since then there will be non-zero frame rotation with an axis of  rotation that is different from $\mathbf e_1$, i.e. not aligned with the field strength (according to \eref{rot}). A co-moving inertial observer will then see a rotating vector.  Note that the vector rotation is a rather unique phenomenon for Bianchi type I, since it is the only space-time that can accommodate a purely spatial field strength with three independent components \cite{thorsrud17}. In fact, as will be shown later, the vector rotation has a crucial role in determining the future asymptotic behavior for all solutions in Bianci type I. It is therefore worthwhile to formally define the vector rotation: 
{\dfn[Vector rotation]\label{def:rotation} The system (\ref{B1_G2:first})-(\ref{B1_G2:last}) with  $V_1\neq0$ possesses a rotating vector if $\Sigma_3\neq0$ and a non-rotating vector if $\Sigma_3=0$.}

In this branch one is thus left with a richer flora of options. Table \eref{tab:FpB1}  gives a summary of equilibrium points. Again, among the solutions are standard solutions without a vector; flat FLRW, Jacobs' Disk (J.D.) and Kasner (K.). Additionally one finds three solutions containing the gauge field:

\begin{itemize}
\item \emph{Wonderland} (W.) is an LRS solution containing both a non-rotating vector and the perfect fluid. The field strength is aligned with the LRS axis and the expansion asymmetry is prolate type. %Since the frame rotation is zero, the direction of the field strength is stable with respect to an inertial observer. 
Its range of existence is the open interval $\gamma\in (2/3 ,2)$. It approaches the flat FLRW solution $\Omega_{\rm pf}$ when $\gamma\rightarrow 2/3$ and the Kasner solution ($\Sigma_+=-1$) when $\gamma\rightarrow 2$. Interestingly, it has a deceleration parameter $q=-1+3\gamma/2$ identical to the flat FLRW solution. The line element of Wonderland is
    \begin{equation}
    {\rm d}s^2=-{\rm d}t^2+t^{2}{\rm d}x^2+t^{\frac{2-\gamma}{\gamma}}({\rm d}y^2+{\rm d}z^2).
    \end{equation}
 
\item \emph{The Rope} (R.) contains a rotating vector and the perfect fluid.  Its range of existence is the open interval $\gamma\in(6/5, 4/3)$. It approaches Wonderland in the limit $\gamma\rightarrow 6/5$ and the Edge in the limit  $\gamma\rightarrow 4/3$. Like Wonderland, it has a deceleration parameter $q=-1+3\gamma/2$ identical to the flat FLRW solution. In the considered $\gamma$ range the shear tensor has three distinct real eigenvalues, but in the limit $\gamma\rightarrow 6/5$ two of them become identical. Thus the Edge is not an LRS solution, although it is ``almost LRS'' close to Wonderland.  %Since $\Sigma_3\neq0 \Rightarrow \Omega_2\neq0$, according to (\ref{rot}), the field strength rotates relative to a co-moving inertial observer. 
The line element of the Rope is
    \begin{equation}\fl
    {\rm d}s^2=-{\rm d}t^2+t^2\left({\rm d}x+ \sqrt{\frac{2(5\gamma-6) }{(2-\gamma)}}t^{1-\frac{2}{\gamma}}{\rm d}z\right)^2+t^{\frac{2(4-3\gamma)}{\gamma}}{\rm d}y^2+t^{\frac{4(\gamma-1)}{\gamma}}{\rm d}z^2.
    \end{equation}
\item \emph{The Edge} (E.) contains only a rotating vector and has deceleration parameter $q=1$, similar to a radiation dominated universe. Since $\,\Omega_{{\rm pf}}=0$, it exists in the entire range of models, $\gamma\in[0,2]$. The (normalized) shear tensor has three distinct eigenvalues $\{ 1, \pm \sqrt{2\mp\sqrt{3}} \}$, so there is no plane of expansion symmetry. %The considered frame in which the field strength is aligned with $\mathbf e_1$ rotates ($\Omega_2\neq0$). One must  conclude that the Edge is a non-LRS solution in which the field strength rotates relative to a co-moving inertial observer. 
The line element of the Edge is
    \begin{equation}
    {\rm d}s^2=-{\rm d}t^2+t^2\left({\rm d}x+\sqrt{2}t^{-1/2}{\rm d}z\right)^2+{\rm d}y^2+t{\rm d}z^2.
    \end{equation}

\end{itemize}

\begin{table}[H]
	\centering
	\resizebox{\textwidth}{!}{\begin{tabular}{lcccccccc}
			\toprule
			\multicolumn{9}{c}{\textbf{Equilibrium sets in Bianchi type I with pure spatial field} ($\Theta=0$)} \\
			\hline
			Name (abbr.) & $q$ & $\gamma $&$\Omega_{\rm pf}$&$\Sigma_+$&$\Sigma_-$&$\Sigma_\times$&$\Sigma_3$&$V_1$\\
			\hline
			Wonderland (W.) & $-1+\frac{3}{2}\gamma $ & $(\frac{2}{3},2)$&$\frac{3}{2}-\frac{3}{4}\gamma $&$\frac{1}{2}-\frac{3}{4}\gamma $&0&0&0& $\frac{3}{4}\sqrt{(\gamma -2)(2-\gamma )}$ \Tstrut \\
			the Rope (R.) & $-1+\frac{3}{2}\gamma $ & $(\frac{6}{5},\frac{4}{3})$&$6-\frac{9}{2}\gamma$&$\frac{1}{2}-\frac{3}{4}\gamma$&$\frac{\sqrt{3}}{4}(6-5\gamma)$&0&$\pm\frac{1}{2}\sqrt{\frac{15}{2}\left(2-\gamma\right)\left(\gamma-\frac{6}{5}\right)}$&$\frac{3}{2}\sqrt{\frac{3}{2}\left(2-\gamma\right)\left(\gamma-\frac{10}{9}\right)}$\\
			the Edge (E.) & $1$& $[0,2]$&0&$-\frac{1}{2}$&$-\frac{1}{2\sqrt{3}}$&0&$\pm\frac{1}{\sqrt{6}}$&$\frac{1}{\sqrt{2}}$\\
			flat FLRW & $-1+\frac{3}{2}\gamma $ & $[0,2)$ &1&0&0&0&0&0\\
			Kasner (K.) & $2$ & $[0,2)$&0&$[-1,1]$& $\pm \sqrt{1-\Sigma_+^2}$ &0&0&0\\
			Jacobs' Disk (J.D.) & $2$ & $2$ &$[0,1]$& $[-1,1]$ & $\pm \sqrt{1-\Omega_{\rm pf}-\Sigma_+^2}$ & 0 & 0& 0 \\
			\hline
	\end{tabular}}
	\caption{Summary of equilibrium points for the branch $\Theta=0$.}
	\label{tab:FpB1}
\end{table}

\subsection{Analysis \label{subs:B1_class}}
In this section the analysis of behaviour for the range of models with $\gamma\in[0,2)$ will be conducted primarily by a local approach, but generalizations to global results established in section \ref{sec:global:BI} will be commented in order to draw the full picture.\setcounter{footnote}{0}\footnote{Note that the considered gamma range is half-open, which excludes Jacobs' Sphere and Jacobs' Disk from the analysis. The particular cases of a stiff fluid ($\gamma=2$) can easily be included in the analysis, using the same methods as below to deal with zero eigenvalues,  but is avoided due to its limited relevance.} The case of a non-rotating vector will be considered before turning the attention to the general, five-dimensional, model.  

\paragraph{A non-rotating vector}
After proper gauge-fixing one can set $\Sigma_3=\Sigma_\times=0$ in the system (\ref{B1_G2:first})-(\ref{B1_G2:last}) in the case of a non-rotating vector, as shown above. In that case the variable $\Sigma_-$ decays monotonically and is thus irrelevant at late times.\footnote{At early times $q\rightarrow 2$ since $\Sigma_-$ is bounded, and thus Kasner is reached asymptotically in the past if $\gamma<2$ or Jacobs' Disk if $\gamma=2$.} The overall picture is clearly the same if one also set $\Sigma_-=0$, i.e. consider the invariant subsystem $\Sigma_-=\Sigma_\times=\Sigma_3=0$. In this case the only shear-variable is $\Sigma_+$ so the subsystem is LRS with an expansion isotropy in the plane orthogonal to the vector. Note that since the subsystem is two-dimensional and invariant, the entire phase flow can be visualized in the plane, as in figure \ref{Fig:phaseportraitBI} for the dust case $\gamma=1$.
\begin{figure}[t!]
	\centering
	\includegraphics[width=0.8\textwidth]{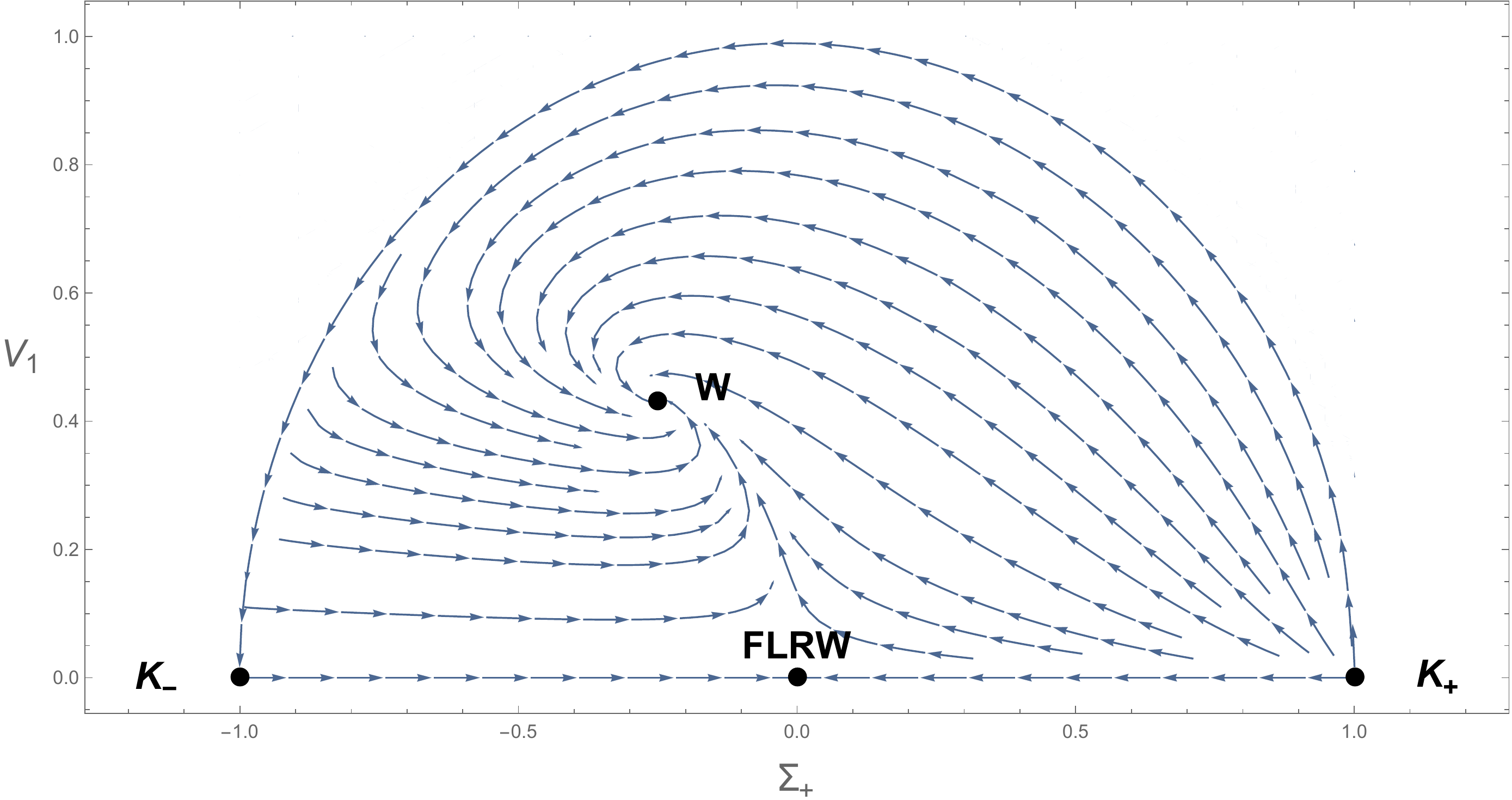}\caption{Phase flow in Bianchi type I for the LRS subsystem $\Sigma_-=\Sigma_\times=\Sigma_3=0$ with $\gamma=1$. Here ``W'' denotes Wonderland and $K_\pm$ denotes Kasner solutions with $\Sigma_+= \pm1$.}
	\label{Fig:phaseportraitBI}
\end{figure}

The equilibrium points and stability can be summarized as follows:
\begin{itemize}
	\item \emph{flat FLRW} has eigenvalues $\{\frac{3}{2} (\gamma-2), -1+\frac{3}{2} \gamma \}$ and is thus stable for $\gamma<2/3$.
	\item \emph{Wonderland} has eigenvalues $-\frac{3}{4}(2-\gamma)(1\pm\sqrt{5-6\gamma})$, whose real parts are negative in the entire range $\gamma\in(2/3, 2)$. Thus Wonderland is an attractor when it exists.
	\item The points $K_\pm$ are the LRS points on the Kasner circle where $\Sigma_+=\pm1$. The eigenvalues of $K_+$ and $K_-$ are $\{4, 6-3\gamma\}$ and $\{0, 6-3\gamma\}$, respectively. The zero eigenvalue of $K_-$ corresponds to a perturbation in the direction of $V_1$. This direction can easily be verified to be stable, so $K_-$ is a saddle. $K_+$ is past stable and the unique repeller of the subsystem.  
\end{itemize}
At the point $\gamma=2/3$ a bifurcation occurs, where Wonderland branches off from the flat FLRW solution and the stability is exchanged, see figure \ref{Fig:stability_BI_nonrotating}. Note that Wonderland is the unique attractor for its full existence range $\gamma\in (\frac{2}{3},2)$ in the presence of a \emph{non-rotating} vector. A stronger global version of this result is inferred from theorem \ref{thm:anisotropic:BI} in the next section. 

\paragraph{General system}
For the general five-dimensional dynamical system the situation is richer. Two more bifurcations in addition to the one at $\gamma=2/3$ occur. First at $\gamma=6/5$, where the Rope branches off Wonderland, and next at $\gamma=4/3$, where the Rope connects with the Edge. The eigenvalues of the linearization matrix around each equilibrium point are summarized in tables  \ref{tab:EigVal_BI_temp} and \ref{tab:EigVal_BI_spatial} for the temporal branch and the spatial branch, respectively. The attractors are determined at linear level for all $\gamma$ values except the bifurcation points and the interval $(\frac{4}{3}, 2]$, where zero eigenvalues are present. For these $\gamma$ values the monotone functions introduced in section \ref{sec:global:BI} has been employed to determine the stability.

All conclusions based on local stability analysis and monotone functions are summarized in table \ref{tab:StabB1}. The diagram in figure \ref{Fig:stability_BI} summarizes the stability at late times. As seen in the figure the attractors are linked in a chain on the gamma range, starting with flat FLRW ($\gamma\in[0,\frac{2}{3}]$), followed by Wonderland ($\gamma\in(\frac{2}{3},\frac{6}{5}]$), the Rope ($\gamma\in(\frac{6}{5},\frac{4}{3})$) and finally the Edge ($\gamma\in[\frac{4}{3},2]$). Thus a unique attractor is identified for each value of $\gamma$. Each of them will be shown to be a \emph{global} attractor in the following subsection. 

Note that for Wonderland the eigenvalue $l_3$ in table \ref{tab:EigVal_BI_spatial}, which becomes positive at the transcritical bifurcation point $\gamma=6/5$, is associated with an eigenvector in the $\Sigma_3$ direction. This explains why the instability of Wonderland for $\gamma\in(\frac{6}{5}, 2)$ is suppressed in the special case of a non-rotating vector ($\Sigma_3=0$). Generally though the vector is expected to rotate for $\gamma>6/5$, since $\Sigma_3\neq0$ in the Rope and the Edge. %These functions have also been employed to prove that Wonderland, the Rope and the Edge in fact are \emph{global} attractors. Thus, Bianchi I universes containing a perfect fluid and a gauge field are generally assymptotically anisotropic with a vector that is rotating if $\gamma > 6/5$. For all cases where the linear stability analysis is conclusive the consistency of the eigenvalues with the monotone functions has been checked. In particular one should note that for Wonderland the eigenvalue $l_3$ in table \ref{tab:EigVal_BI_spatial}, which becomes positive at the transcritical bifurcation point $\gamma=6/5$, is associated with an eigenvector in the $\Sigma_3$ direction. This explains why the instability of Wonderland for $\gamma\in(\frac{6}{5}, 2)$ is suppressed in the special case of a non-rotating vector ($\Sigma_3=0$).

For the past stability qualitative insight can be obtained directly by applying the monotone functions collected in \ref{App:Zs}. For example, the function $Z_1$ is given by
\begin{eqnarray}
Z_1&=\frac{\Sigma_\times\Sigma_3^2V_1^3}{\Omega^3_{\rm pf}}, \qquad Z_1'=3(3\gamma-4)Z_1. 
\label{z1}
\end{eqnarray}
For $\gamma<4/3$, $Z_1$ is monotonically decreasing, and since all variables are bounded, $\Omega_{\rm pf}\rightarrow 0$ at early times. Comparing with table \ref{tab:FpB1} one must conclude that for the vector branch ($\Theta=0$), the only fix point candidate for a past attractor is Kasner. Under the weak assumption that the system is not chaotic at early times, one is lead to conclude that points, or domains, on the Kasner circle serve as global repellers of the system, for the wide range $\gamma<4/3$. To determine exactly which parts of the Kasner circle that has this role the attention is next turned to the local analysis. 

%For equilibrium points on the boundary between the two branches, i.e. those without a gauge field (flat FLRW and Kasner), the union of the two tables have been used in order to determine the stability with respect to all possible directions in the Bianchi type I model. 

At first sight the past stability analysis is complicated by the presence of multiple zero eigenvalues in the relevant equlibrium points. However, several of these merely reflects the dimensionality of the equilibrium sets. For instance Jacobs' Extended Disk (J.E.D.) is a two-parameter family of fixed points and thus the two zero eigenvalues in table \ref{tab:EigVal_BI_temp} merely reflect perturbations along the equilibrium set itself. The Kasner circle is the subset $\Theta=0$ of J.E.D. and thus there are two zero eigenvalues in table \ref{tab:EigVal_BI_temp} for the temporal branch and one zero eigenvalue in table \ref{tab:EigVal_BI_spatial} for the spatial branch. %To see this note that Kasner is the boundary unit circle in the $\Sigma_+$-$\Sigma_-$ plane and J.E.D. is the interior of the disk if $\Theta$ is eliminated by the constraint equation. 
Since it is natural and more useful to classify the stability of the equlibrium set as a whole, rather than for its individual points, the zero eigenvalues of Kasner and J.E.D. in table \ref{tab:EigVal_BI_temp} are safely ignored. With these prescriptions the stability of Kasner and J.E.D. ($\Theta\neq0$) is given in table \ref{tab:StabB1}. Note that J.E.D. with $\Theta\neq0$ is a repeller for all values of $\gamma<2$. For the vector branch it is only the subset $\Theta=0$ of J.E.D. that is dynamically relevant. Since Kasner is on the boundary between the temporal and spatial branches, the union of the eigenvalues in tables \ref{tab:EigVal_BI_temp} and \ref{tab:EigVal_BI_spatial} have been used in the classification of stability to account for all directions in Bianchi type I. The Kasner circle is divided into well defined repeller and saddle domains.\setcounter{footnote}{0}\footnote{At the Kasner points $\theta=\pi$ and $\theta=4\pi/3$ there are additional zero eigenvalues not accounted for by the dimensionality of the equilibrium set. However, the former point is a saddle in the LRS subsystem and thus also in the general system, whereas the latter point remains inconclusive at linear order.} In polar coordinates $(\Sigma_+, \Sigma_-)=(\cos \theta, \sin \theta)$ the part $\theta \in (\pi, \frac{4\pi}{3})$ is past stable, the point $\theta=4\pi/3$ is inconclusive  and the rest of the Kasner circle are saddle points. 

\begin{figure}[t!]
	\centering	
	\begin{overpic}[width=0.9\textwidth,tics=10]{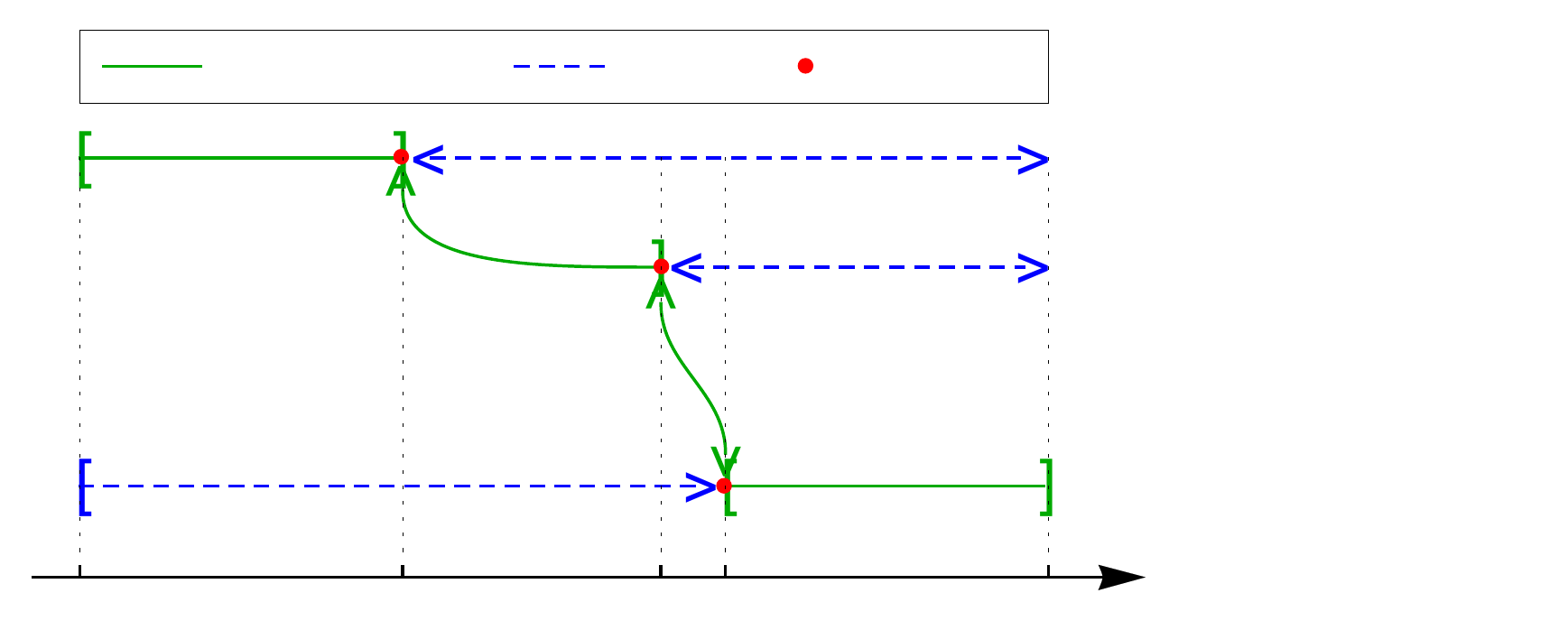}		
		\put (14,35) {\scriptsize global attractor}
		\put (40,35) {\scriptsize saddle}
		\put (53,35) {\scriptsize bifurcation}		
		\put (70,29) {\footnotesize flat FLRW, $\displaystyle \gamma\in[0,2)$}
		\put (70,22) {\footnotesize Wonderland (W.), $\displaystyle \gamma\in\left(2/3,2\right)$}
		\put (70,16) {\footnotesize the Rope (R.), $\displaystyle \gamma\in\left(6/5,4/3\right)$}		
		\put (70,9) {\footnotesize the Edge (E.), $\displaystyle \gamma\in\left[0, 2\right]$}	
		\put (74,3) {\large $\displaystyle \gamma$ }	
		\put (3,0) { \scriptsize $\displaystyle 0$ }
		\put (25,-1) {\scriptsize $\displaystyle \frac{2}{3}$ }	
        \put (41,-1) {\scriptsize $\displaystyle \frac{6}{5}$ }	
		\put (45,-1) {\scriptsize $\displaystyle \frac{4}{3}$ }
		\put (65,0) { \scriptsize $\displaystyle 2$ }		
	\end{overpic}	
	\caption{Stability diagram for Bianchi type I. The stability of flat FLRW, W., R. and E. in their full $\gamma$-ranges are indicated by solid green lines (global attractor), dashed blue lines (saddle). Red dots represent bifurcation points where the stability is exchanged.  Closed and open intervals are indicated using $[$, $]$ and $<$, $>$, respectively.  } 
	\label{Fig:stability_BI}
\end{figure}

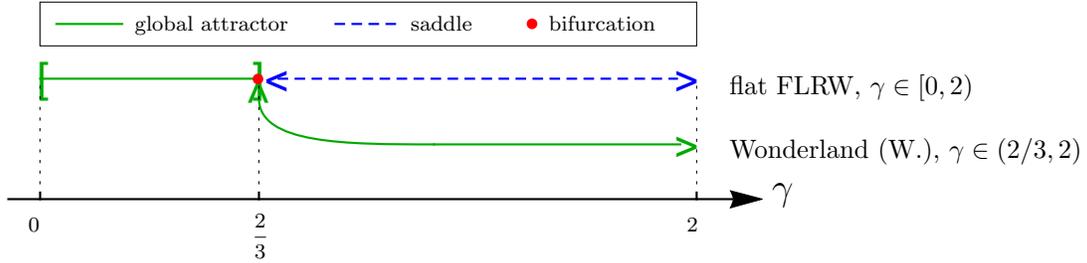
\begin{figure}[t!]
	\centering	
	\begin{overpic}[width=0.9\textwidth,tics=10]{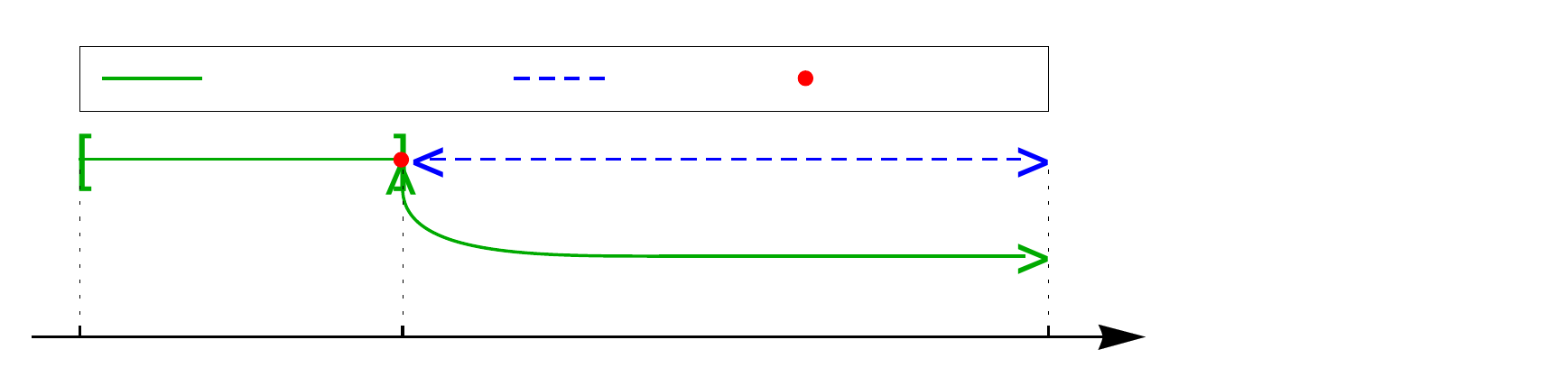}		
		\put (14,19) {\scriptsize global attractor}
		\put (40,19) {\scriptsize saddle}
		\put (53,19) {\scriptsize bifurcation}		
		\put (70,13) {\footnotesize flat FLRW, $\displaystyle \gamma\in[0,2)$}
		\put (70,7) {\footnotesize Wonderland (W.), $\displaystyle \gamma\in\left(2/3,2\right)$}	
		\put (74,3) {\large $\displaystyle \gamma$ }	
		\put (3,0) { \scriptsize $\displaystyle 0$ }
		\put (25,-1) {\scriptsize $\displaystyle \frac{2}{3}$ }		
		\put (65,0) { \scriptsize $\displaystyle 2$ }		
	\end{overpic}	
	\caption{Stability diagram for the invariant subspace $\Sigma_3=0$ of Bianchi type I. See caption of figure \ref{Fig:stability_BI} for further details.} 
	\label{Fig:stability_BI_nonrotating}
\end{figure}

\begin{table}[H]
	\centering
	\resizebox{\textwidth}{!}{\begin{tabular}{llllclclclc}
			\toprule
			\multicolumn{9}{c}{\textbf{Classification of equilibrium sets in Bianchi type I}} \\
			\hline
			Name (abbr.)  && Existence  && Attractor && Saddle && Repeller  \\
			\hline
			Wonderland (W.) && $\gamma \in (\frac{2}{3},2)$ &&$\gamma\,\in\,(\frac{2}{3},\frac{6}{5}]$&&$\gamma\,\in\,(\frac{6}{5},2)$&& \Tstrut \\
			the Rope (R.)&&$\gamma \in (\frac{6}{5},\frac{4}{3})$ &&$\gamma\,\in\,(\frac{6}{5},\frac{4}{3})$&& &&\\
			the Edge (E.)&&$\gamma \in [0,2]$ && $\gamma\,\in\,[\frac{4}{3}, 2]$ &&$\gamma\,\in\,[0,\frac{4}{3})$&& \\
			flat FLRW&&$\gamma \in [0,2)$ && $\gamma\,\in\,[0,\frac{2}{3}]$ &&$\gamma\,\in\,(\frac{2}{3},2)$&& \\
			%			Kasner (K.) && $\gamma \in [0,2)$ &&&&   $\Sigma_->0 \cup \Sigma_- < \sqrt{3} \Sigma_+$ && && else \\ 
			Kasner (K.) &&$\gamma \in [0,2) $ && && $\Sigma_+>\frac{\Sigma_-}{\sqrt{3}}\,\cup\,\Sigma_->0$ && $-1<\Sigma_+<\frac{\Sigma_-}{\sqrt{3}}<0$ \\
			J.E.D. with $\Theta\neq0$ &&$\gamma \in [0,2) $ && &&  && $\forall$  
	\end{tabular}}
	\caption{The existence domain of each equilibrium set is divided into attractor, saddle and repeller subdomains. Results are obtained using monotone functions as well as standard local analysis. }
\label{tab:StabB1}
\end{table}

%\paragraph{Jacobs' sphere} If $\Theta\neq0$ one only needs to look at the eigenvalues in table \ref{tab:EigVal_BI_temp} which are all zero and thus inconclusive. For $\Theta=0$ one must use the union of the eigenvalues in table \ref{tab:EigVal_BI_temp} and \ref{tab:EigVal_BI_spatial} (called \emph{Jacobs' disk} in the latter, since $\Theta=0$), in order to determine the stability in all perturbation directions possible in Bianchi type I. Due to zero eigenvalues, the only possibility for a conclusive classification is that there are both positive and negative signs among $\{l_3, l_4, l_5\}$ in table \ref{tab:EigVal_BI_spatial}. There are two possibilities: 
%\begin{itemize}
%	\item 	$\Theta=0\cap \Sigma_+>0 $ \newline
%	This corresponds to $l_3<0$, $l_4<0$, $l_5>0$
%	\item 	$\Theta=0\cap \Sigma_+<0 \cap ( \Sigma_-<\sqrt{3} \Sigma_+ \cup \Sigma_->0 ) $ \newline
%	For $\Sigma_+\in(-1,0)$ this corresponds to $l_5>0$ and either $l_3$ or $l_4$ negative. For $\Sigma_+=-1$ this corresponds to $l_3<0$, $l_4>0$ and $l_5=0$.
%\end{itemize}

\subsection{Global attractors and anisotropic hair theorems \label{sec:global:BI}}

Note that the presence of a unique attractor point in state space does not, by itself, imply that the corresponding self-similar solution is the global future asymptote. The reason, of course, is that instead of ending in a point, there is the possibility that the orbit ends in a closed loop.  Monotone functions provide a powerful tool for identifying \emph{global} attractors. The reason is the invariance of initial conditions; their absolute value is non-decreasing, or non-increasing, in all of state space. In fact, the set (\ref{B1_G2:first})-(\ref{B1_G2:last}) possesses a number of monotone functions, summarized in \ref{App:Zs}. Equipped with these functions the locally stable equilibrium points discussed above, i.e. flat FLRW, Wonderland, the Rope and the Edge, will be identified as global attractors. Furthermore, this also gives the stability at the bifurcation points and $\gamma>4/3$, where the linear approach is inconclusive, without the need of centre manifold analyses. The main consequences of the monotone functions are summarized in the following theorems. 

The first monotone function to be considered is $\Omega_{\rm pf}$, which is monotonically increasing in the closed interval $\gamma\in[0, 2/3]$, according to (\ref{JacId2}). This is the case in all Bianchi type I-VII$_{h}$ models, which led to the no-hair theorem \ref{NoHair} for the half-open interval $\gamma\in[0, 2/3)$. It is easy to find counter examples to the no-hair theorem for $\gamma=2/3$; in that case the open FLRW universe with $\Omega_{\rm pf}\in(0,1)$ is a self-similar solution in Bianchi type V.  Because of its special role as a bifurcation point in the Bianchi type I model, it is worth to note that the no-hair theorem can be extended to include the point $\gamma=2/3$ for this particular space-time:
 
\begin{thm}[No-hair theorem for  $0\le\gamma\le2/3 $]
	\label{th:BI1}
	Bianchi type I with a $j$-form, a perfect fluid $\Omega_{\rm pf}$ with equation of state parameter $0\leq\,\gamma\,\leq\,2/3$ will be asymptotically flat FLRW with $\Omega_{\rm pf}=1$ and $q=\frac{3}{2}\gamma-1\,\le\,0$.
\end{thm}
\begin{proof}
	The half-open range $\gamma\in[0, 2/3)$ is the Bianchi type I special case of the no-hair theorem \ref{NoHair}. In the case $\gamma=2/3$ the equation of motion for $\Omega_{\rm pf}$ can be written 
	\begin{equation}
	(\log\Omega_{\rm pf})' = 2 q \ge 0.
	\end{equation}
	Note that $q\ge0$ in both branches of Bianchi type I. Since $\Omega_{\rm pf}$ is monotonically increasing and bounded it follows that 
	\begin{equation}
	\lim_{\tau\rightarrow\infty} q=0
	\label{heisann1}
	\end{equation}
	In the perfect branch one has $q=2(\Sigma_+^2+\Sigma_-^2+\Theta^2)$ so that (\ref{heisann1}) and the constraint (\ref{BI:con1}) directly gives
	\begin{equation}
	\lim_{\tau\rightarrow\infty} \Omega_{\rm pf}=1.
	\end{equation}  
	In the vector branch one has $q=2\Sigma^2 = 2(\Sigma_+^2+\Sigma_-^2+\Sigma_\times^2+\Sigma_3^2)$ so that $\Omega_{\rm pf}+V_1^2=1$ at late times according to (\ref{heisann1}) and the constraint (\ref{HamB1V}). If $\lim_{\tau\rightarrow\infty} V_1^2>0$ (\ref{B1_G2:first}) gives $\lim_{\tau\rightarrow\,\infty}\left(\Sigma_+'\right)\neq 0$ which contradicts the fact that $\Sigma^2=0$ according to (\ref{heisann1}). The only possibility is therefore again $\Omega_{\rm pf}=1$ at late times.	 
\end{proof}
As for the case $2/3<\gamma<2$ there are  anisotropic attractors and  a number of monotone functions that determine the global behaviour at late times, see \ref{App:Zs}. Some general observations for these functions are that they are of the form $Z=x^ay^bz^c/\phi^d$, for positive (or zero) exponents $a, b, c, d$. In addition, $Z'>0$, unless a number of the variables are equal to their fixed-point values. Furthermore, for the corresponding fix points, the variables $x,y,z$ in $Z$ are non-zero. This means that unless $x\equiv 0$ etc., $x$ cannot be zero at any time later, nor can $x\rightarrow 0$ since all variables $x,y,z$ are bounded and $\phi$ cannot be zero (see \cite{Hervik:2010uh}). Let $(V_1)_0$ denote the value of $V_1$ at the equilibrium point. The above observations imply  the following result: 
\begin{thm}[Anisotropic hairs for $2/3<\gamma<2$]
Assume that $\Omega_{{\rm pf}}, V_1>0$ and $2/3<\gamma<2$. Then a Bianchi type I with a spatial $j$-form will be asymptotically anisotropic, with $V_1\rightarrow (V_1)_0\neq 0$ at late times. More specifically, if: 
\begin{itemize}
\item{} $2/3<\gamma<2$ and $\Sigma_3=0$, then it will  asymptotically approach \emph{Wonderland}; 
\item{} $2/3<\gamma\leq 6/5$ and $\Sigma_3\neq 0$, then it will  asymptotically approach \emph{Wonderland}; 
\item{} $6/5<\gamma< 4/3$ and $\Sigma_3\neq 0$, then it will  asymptotically approach \emph{the Rope}; or
\item{} $4/3\leq \gamma< 2$ and $\Sigma_3\neq 0$, then it will asymptotically approach \emph{the Edge}.
\end{itemize}
\label{thm:anisotropic:BI} 
 \end{thm}
 \begin{proof}
 Use monotonic functions $Z_{2}$, $Z_3$ and $Z_4$. 
\end{proof}
This theorem implies that the $j$-form is dynamically significant when $\gamma>2/3$ and that it is asymptotically self-similar with non-zero shear. Furthermore, only if $\Sigma_3=0$ exactly, or if $\gamma\in (2/3,6/5]$, the space-time is asymptotically LRS. Hence, in the case of radiation ($\gamma=4/3$), the space-time is generically non-LRS, even asymptotically non-LRS. Recall that $\Sigma_3=0$ amounts to a non-rotating vector and $\Sigma_3\neq0$ to a rotating vector (definition \ref{def:rotation}). The vector rotation is thus identified as a cruical physical mechanism that controls the asymptotic behavior. Generally vector rotation is expected for $\gamma>6/5$ and is thus relevant in the early radiation dominated universe in the presence of gauge modes at large scales.

%%%%%%%%%%%%%%%%%%%%%%
\section{Dynamical system in Bianchi type V}
\label{sec:dyn5}
%%%%%%%%%%%%%%%%%%%%%

As noted from table \ref{tab:J} the type V space-time is specified by
\begin{equation}
\label{BinachiV}
\vert\mathbf{N}_\Delta\vert=N_+=0\phantom{000},\phantom{000}A\,\neq\,0.
\end{equation}
The constraints \eref{Constr1} and \eref{Constr4} consequently give $\mathbf{V}_c=0$ and $\mathbf{\Sigma}_{1}=0$. By theorem \eref{thmDeSitter} the cosmological constant will eventually come to dominate if present. The case considered will therefore be $\Omega_\Lambda=0$, and the complete set of evolution equations \eref{FluidEqs}-\eref{JacId2} reduces to 
\begin{eqnarray}
\label{BVDelta}
&\mathbf{\Sigma}_\Delta'=(q-2-2iR_1)\mathbf{\Sigma}_\Delta,\\
&\Sigma_+'=\left(q-2\right)\Sigma_+-2V_1^2,\\
&\Theta'=(q-2)\Theta-2AV_1,\\
&A'=\left(q+2\Sigma_+\right)A,\\
&V_1'=\left(q+2\Sigma_+\right)V_1,\\
&\Omega_{\rm pf}'=2\left(q+1-\frac{3}{2}\gamma \right)\Omega_{\rm pf}\,.
\end{eqnarray}
From \eref{Constr2} and \eref{Constr3} these equations are subject to the constraints
\begin{eqnarray}
&1=\Sigma_{+}^2+\abs{\mathbf{\Sigma}_{\Delta}}^2+\Omega_{\rm pf}+\Theta^2+V_1^2+A^2, \label{ConstrBV:1}\\
&\Theta V_1=A\Sigma_+\,. \label{ConstrBV:2}
\end{eqnarray}
Among the 5 independent degrees of freedom that are left after accounting for the 2 constraints, one merely reflects the gauge freedom associated with rotations of the $G_2$ plane. Before carrying out a dynamical systems analysis the frame will be specified uniquely so that the final type V system becomes 4 dimensional, i.e. 4 independent physical degrees of freedom.

\subsection{Gauge fixing}

According to (\ref{trans:1}), $\mathbf{\Sigma}_\Delta$ is the only non-scalar quantity above. It transforms as a spin-2 object under a rotation of frame with respect to the basis vector $\mathbf e_1$. From the above equations it is evident that working in F-gauge ($R_1=\phi'=0$) is advantageous. The shear is already almost diagonal. In the F-gauge it can be fully diagonalized by choosing an initial orientation ($\phi$) for the tetrad such that $\mathbf{\Sigma}_\Delta$ is purely real. It is then clear from the evolution equation (\ref{BVDelta}) that the shear will remain diagonal. Note that the F-gauge, in the case of  type V, corresponds to an inertial frame following gyroscopes. This must be so since $R_1=0\,\rightarrow\,\Omega_1=0$ and since $\mathbf{\Sigma}_1=0$ implies $\Omega_2=\Omega_3=0$ by \eref{rot}.

To summarize, the gauge choice is implemented by the following replacements in equation (\ref{BVDelta}):
\begin{equation}
\mathbf{\Sigma}_\Delta \rightarrow \Sigma_-, \qquad R_1 \rightarrow 0.
\end{equation}
Note that the final physical system is 4 dimensional, due to the constraints (\ref{ConstrBV:1}) and (\ref{ConstrBV:2}) and it is mathematical equivalent to the one obtained by instead replacing $\mathbf{\Sigma}_\Delta$ by the gauge independent quantity $\Delta\equiv(\mathbf{\Sigma}_\Delta \mathbf{\Sigma}_\Delta^*)^{1/2}$.  

\subsection{Reduced system}
The dynamical system possesses some  useful properties. Firstly, by computing the derivative of the ratio $V_1/A$, it is seen that it is constant: $(V_1/A)'=0$. In addition, by solving the constraint equation (\ref{ConstrBV:2}), new variables $\eta$, $\nu$ and $\alpha$ may be introduced:
\begin{equation}
(A,V_1)=\eta(\cos\alpha, \sin \alpha), \qquad (\Theta,\Sigma_+)=\nu (\cos\alpha, \sin\alpha). 
\end{equation}
\begin{eqnarray}
\label{BVRed}
&\Sigma_-'=(q-2){\Sigma}_-,\\
&\nu'=\left(q-2\right)\nu-2\eta^2\sin\alpha,\\
&\eta'=\left(q+2\nu\sin\alpha\right)\eta,\\
&\alpha'=0,\\
&\Omega_{\rm pf}'=2\left(q+1-\frac{3}{2}\gamma \right)\Omega_{\rm pf}\,. \label{BVRedlast}
\end{eqnarray}
From \eref{Constr2} and \eref{Constr3} these equations are subject to the constraints
\begin{eqnarray}
&1={\Sigma}_-^2+\Omega_{\rm pf}+\eta^2+\nu^2, \label{ConstrBV:Red1}
\end{eqnarray}
Clearly, every value of $\alpha$ gives an invariant subspace and thus the analysis can be reduced accordingly. Note that the quantities $(V_1/A)$ and $(\Sigma_+/\Theta)$ are equal and constant ($=\tan\alpha$). Moreover, $Z_5=\Sigma_-$ is a monotone function which for all $\gamma<2$ decays towards $0$. The above system can easily be analysed. 

\subsection{Equilibrium sets}
The list of equilibrium points found in type V is given in table \ref{tab:FpBV}. Among the solutions are again Jacobs' Extended Disk, Kasner and Jacobs' Sphere, first presented in section \ref{ch:puretemporal}. Since type V has maximally symmetric hypersurfaces of homogeneity with negative curvature, the open and flat FLRW universes are also among the solutions.

Self-similar solutions with isotropy-violating matter sector ($V_1\neq 0$) fall within two one-parameter families of equilibrium points. One of them is the Plane Wave type solutions, which are solutions of the Einstein Field Equations possessing a covariantly constant null vector. In general plane wave solutions have a 2-dimensional Abelian group acting freely on the spatial hypersurfaces. That is to say: they are special exact $G_2$ solutions to the evolution equations, as confirmed in the present analysis~\cite{gron07}. Confer with \cite{hervik03} for more on plane waves. The second equilibrium set is a Bianchi type V generalization of Wonderland. Both sets are locally rotationally symmetric (LRS) solutions, since the spatial sections have isotropic geometry and the shear tensor is axisymmetric and aligned with the spatial part of the field strength $j$-form. Since $\Sigma_+<0$ in both sets, the expansion anisotropy is prolate type, i.e. a spherical initial configuration of test particles will be deformed to a prolate spheroid. Further details on the geometry of these solutions:
\begin{itemize}
	\item \emph{Plane Wave}  is a one-parameter family $\mathcal P_{\rm P.W.}(\Sigma_+)$ of equilibrium points with $\Sigma_+\in(-1,0]$. The \emph{Milne solution} corresponds to the point $\Sigma_+=0$. In the other end, the Plane Waves approach the Kasner vacuum solution as $\Sigma_+\rightarrow-1$. Using $\{\Sigma_-, \Sigma_+, A, V_1\}$ as independent variables $\mathcal P_{\rm P.W.}$ is a helix in state space, whose projection on the $\Sigma_+$-$V_1$ plane is the circle centered at $(-\frac{1}{2},0)$ with radius $1/2$. Thus there are two P.W. points for each value of $\Sigma_+$, as seen in table \ref{tab:StabB5}. Although not one-to-one with the curve, $\Sigma_+$ is used as parameter since its value controls the stability uniquely, as it turns out. 
	The line-element is
	\begin{equation} {\rm d} s^2=-{\rm d}t^2+t^2{\rm d}x^2+t^{2s}e^{2sx}\left({\rm d}y^2+{\rm d}z^2\right), 
	\end{equation}   
	where $0<s<1$. Here, the parameter $\Sigma_+$ is given by $\Sigma_+=-(1-s)/(1+2s)$. 
	
	\item \emph{Wonderland} is a one-parameter family $\mathcal P_{\rm W}(\lambda)$ of equilibrium points, where $\lambda$ is a free parameter that is restricted to the range 
 \begin{equation}
    0\,\leq\,\lambda\,< \, \lambda_{\rm sup} \equiv \frac{\sqrt{3}}{4}\sqrt{2-\gamma}.
\end{equation}
Wonderland joins the Plane Wave set in the limit $\lambda \rightarrow\lambda_{\rm sup}$ (see figure \ref{Fig:stability_BV}). It is a straight line segment in state space, in fact a \emph{chord} of the Plane Wave curve. This line is a Bianchi type V generalization of the Bianchi type I fix point (with the same name) studied in section \ref{sec:dyn1}, that corresponds to the point $\lambda=0$.  Its range of existence is the open interval $\gamma\in (2/3 ,2)$. It approaches the open FLRW solution when $\gamma\rightarrow 2/3$ (in which case $\lambda=0$ corresponds to flat FLRW and $\lambda\rightarrow\lambda_{\rm sup}$ corresponds to Milne) and the Kasner solution ($\Sigma_+=-1$) when $\gamma\rightarrow 2$. Interestingly, the entire curve has a deceleration parameter $q=-1+3\gamma/2$ identical to the flat FLRW solution. The line-element of the type V Wonderland is
	\begin{equation} 
    {\rm d} s^2=-{\rm d}t^2+t^2{\rm d}x^2+t^{\frac{2-\gamma}{\gamma}}e^{2kx}\left({\rm d}y^2+{\rm d}z^2\right),
    \end{equation}
	where $0<k< \frac{2-\gamma}{2\gamma}$. Explicitly, the parameter $k$ is related to $\lambda$ via $k=2\sqrt{3(2-\gamma)}\lambda/(3\gamma)$.	
\end{itemize}
The relation among the self-similar solutions as a function of the model parameter $\gamma$ is portrayed in figure \ref{fig:final}.
\begin{figure}[H]
	\centering
	\includegraphics[width=\textwidth]{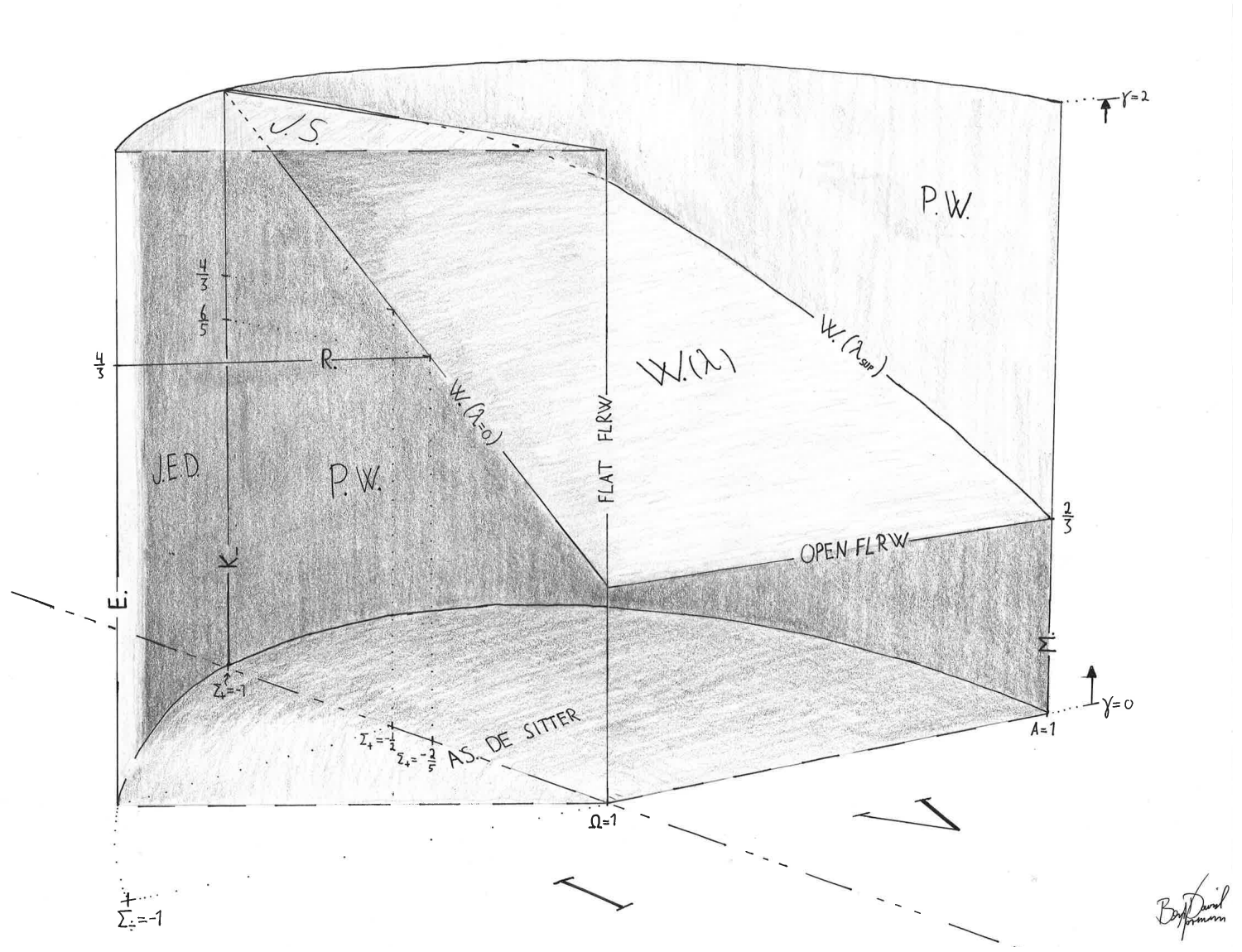}\caption{An artist's impression of the Bianchi type I and type V state spaces as a function of $\gamma$. It serves its purpose of giving a correct overview of how the solutions are all connected to each other as the parameter $\gamma$ varies (from bottom to top). Solid lines and labeled planes represent self-similar solutions, abbreviations can be looked in tables \ref{tab:FpB1} and \ref{tab:FpBV}.}  
\label{fig:final}
\end{figure}

\subsection{Analysis}
In this section the structure of state space will be analyzed based on local stability analysis and geometric observations. Generalizations to global results obtained in the following subsection will be commented on. The eigenvalues of the linearization matrix around each fixed point can be found in table \ref{tab:FpBV} and the stability of each equilibrium set is summarized in table \ref{tab:StabB5}.
\begin{table}[H]
	\centering
	\resizebox{\textwidth}{!}{\begin{tabular}{clcccccccc}
			\toprule
			\multicolumn{10}{c}{\textbf{Equilibrium sets in Bianchi type V}} \\
			\hline
			Type&Name (abbr.) & $q$ &$\gamma $&$A$&$\Omega_{\rm pf}$&$\Sigma_+$&$\Sigma_-$&$V_1$&$\Theta$\\
			\hline
			V&Plane wave (P.W.) & $-2\Sigma_+$ &$[0,2]$&$(\Sigma_++1)$&0&$(-1, 0]$&0&$\pm\sqrt{-\Sigma_+(\Sigma_++1)}$& $\mp\sqrt{-\Sigma_+(\Sigma_++1)}$ \Tstrut \\
			V& Wonderland (W.) & $-1+\frac{3}{2}\gamma $ &$(\frac{2}{3},2)$&$\sqrt{3(2-\gamma )}\,\lambda$&$\frac{3}{4}\left(2-\gamma \right)-4\lambda^2$&$\frac{1}{2}-\frac{3}{4}\gamma $&0&$\pm\frac{\sqrt{3}}{4}\sqrt{(3\gamma -2)(2-\gamma )}$&$\mp\sqrt{3\gamma - 2}\,\lambda$\\
			V&open FLRW & $0$&$\frac{2}{3}$&$(0,1)$&$\sqrt{1-A^2}$&0&0&0&0\\
			I&flat FLRW & $-1+\frac{3}{2}\gamma $&$[0,2)$&0&1&0&0&0&0\\
			%I&Kasner (K.) & $2$ & $[0,2)$&0&0&$[-1, 1]$&$\pm \sqrt{1-\Sigma_+^2}$&0&0 \\
			I& Jacobs' Ext. D. (J.E.D.) &  $2$&$[0,2)$&0&0&$[-1,1]$&$[-1,1]$&0& $\pm \sqrt{1-\Sigma_+^2 - \Sigma_-^2}$ \\
			I & Jacobs' Sphere (J.S.) & $2$ & $2$ &0& $[0,1]$ & $[-1,1]$ & $[-1,1]$ & 0& $\pm \sqrt{1-\Sigma_+^2 - \Sigma_-^2-\Omega_{\rm pf}}$ \\
			\hline
	\end{tabular}}
	\caption{Summary of equilibrium points for Bianchi type V model. Note that the Milne universe (not explicitly in the table) is a subset ($\Sigma_+=0$) of P.W. . Similarly the Kasner (K.) vacuum solution is a subset ($\Theta=0$) of J.E.D. .}
	\label{tab:FpBV}
\end{table}

\begin{table}[H]
	\centering
	\resizebox{\textwidth}{!}{\begin{tabular}{clcclclclclc}
			\toprule
			\multicolumn{12}{c}{\textbf{Classification of equilibrium sets in Bianchi type V}} \\
			\hline
			Type&Name (abbr.) && Existence && Attractor && Saddle && Repeller && Inconclusive \\
			\hline 
			V&Plane Wave (P.W.) && $\gamma\in[0,2], \Sigma_+\in(-1,0]$ &&$\gamma>\frac{2}{3} \cap \Sigma_+>\frac{1}{2}-\frac{3}{4}\gamma$ && else &&&&$\gamma\,=\,\frac{2}{3} \cap \Sigma_+=0$ \Tstrut  \\
			V&Wonderland (W.) && $\gamma\in(\frac{2}{3},2)$, $\lambda\in[0,\lambda_{\rm sup})$ && $\forall$  && &&&&  \\
			V&open FLRW  && $\gamma = \frac{2}{3}$ &&&&&&&&$\forall$\\
			I&flat FLRW && $\gamma\in[0,2)$ &&$\gamma\,\in\,[0,\frac{2}{3})$&&$\gamma\,\in\,(\frac{2}{3},2)$&&&&$\gamma=\frac{2}{3}$\\
			%I&Kasner (K.) && $\gamma\in[0,2)$ &&&&&& $\Sigma_+>-1$ &&$\Sigma_+\,=-1$\\
			I&Jacobs' Ext. D. (J.E.D.) && $\gamma\in[0,2)$ &&&&&& $\Sigma_+>-1$ &&$\Sigma_+\,=-1$\\
			I&Jacobs' Sphere (J.S.) && $\gamma=2$ &&&& && $\Sigma_+>-1$ &&$\Sigma_+\,=-1$ \\ \hline
	\end{tabular}}
	\caption{The domains where the stability analysis is conclusive are divided into attractor, saddle and repeller subdomains by the conditions above.  The rightmost column shows the domains where the linear stability analysis is inconclusive. }
	\label{tab:StabB5}
\end{table}

Several of the equilibrium sets, specifically P.W, W, J.E.D. and J.S., are curves or multidimensional regions in state space.  As in the analysis of the Bianchi type I model, the emphasize is on classifying the stability of each equilibrium set as a whole. Zero-eigenvalues that corresponds to a perturbation along the equilibrium set itself is therefore again ignored in the stability analysis (see last paragraph of section \ref{subs:B1_class}). In each case of a zero eigenvalue the direction of the corresponding eigenvector has been checked explicitly. For instance the Plane Wave set is the set of points:
\[
\mathcal P_{\rm P.W.}(\Sigma_+) = (0, \Sigma_+, -(\Sigma_++1), -\sqrt{-\Sigma_+(\Sigma_++1)})
\]
using $\{\Sigma_-, \Sigma_+, A, V_1\}$ as independent variables. It is easy to verify that the tangent vector $\frac{\mathcal P_{\rm P.W.}}{d \Sigma_+}$ is an eigenvector of the linearization matrix around P.W. Specifically it corresponds to the eigenvalue $l_3=0$ in table \ref{tab:FpBV} and hence corresponds to a perturbation along the equilibrium line $\mathcal P_{\rm P.W.}(\Sigma_+)$. The classification of P.W. in table \ref{tab:StabB5} is therefore based solely on the three eigenvalues $\{l_1,l_2,l_4\}$, i.e. on linear perturbations normal to the curve $\mathcal P_{\rm P.W.}(\Sigma_+)$. By the same argument, the stability of Wonderland is based solely on the three eigenvalues $\{l_1,l_2,l_4\}$, since $l_3=0$ correponds to a perturbation in the curvature direction ($A$), i.e. along $\mathcal P_{\rm W}(\lambda)$. 

For $\gamma<2/3$ the flat FLRW solution is stable and the unique attractor, in consistency with the no-hair theorem \ref{NoHair}. As for the past stability, all of J.E.D. is a repeller, including the Kasner subset $\Theta=0$ (apart from the point $\Sigma_+=-1$ which is inconclusive at linear order).  
%Kommenter figur.
%\begin{figure}[t!]
%	\centering
%	\includegraphics[width=0.5\textwidth]{simulation_inflation.pdf}
%	\caption{Simulation of the Bianchi type V subsystem $\Sigma_-=0$ with $\gamma=0.1$. Integral curves are in red and future direction indicated by arrows.}
%	\label{Fig:sim1}
%\end{figure}

For $\gamma>2/3$ the structure of state space is rather intruiging and the relation among the attractors is shown in figure \ref{Fig:stability_BV}. The part of $\mathcal P_{\rm P.W.}(\Sigma_+)$ where $\Sigma_+>\frac{1}{2}-\frac{3}{4}\gamma$ is stable, whereas the part $\Sigma_+<\frac{1}{2}-\frac{3}{4}\gamma$ consists of saddle points. The point $\Sigma_+ = \frac{1}{2}-\frac{3}{4}\gamma$, where the eigenvalue $l_4$ vanishes, happens to be the point where Wonderland joins Plane Wave (see figure \ref{Fig:stability_BV}). One must conclude that the Plane Wave at this particular point is unstable and a saddle since a perturbation in the direction corresponding to $l_4$ gives a Wonderland solution. The Wonderland curve $\mathcal P_{\rm W}(\lambda)$ on the other hand is stable for all $\lambda$. 
\begin{figure}[t!]
	\centering
	\begin{overpic}[width=0.5\textwidth,tics=10]{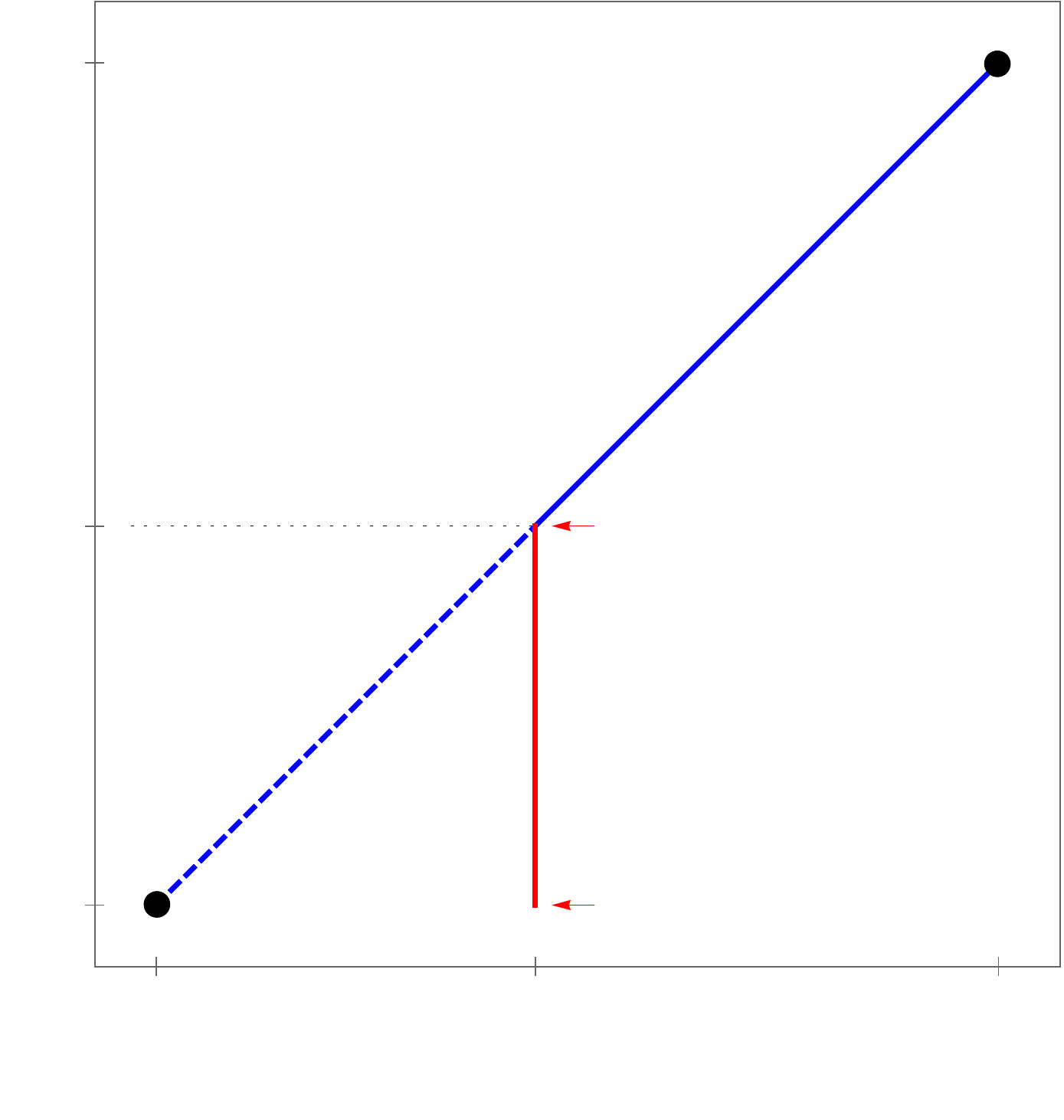}
		
		\put (38,49) {\rotatebox{45}{$\displaystyle \color{blue}\mathcal{P}_{\rm P.W.}(\Sigma_+)$}}
		
		\put (38,49) {\rotatebox{45}{$\displaystyle \color{blue}\mathcal{P}_{\rm P.W.}(\Sigma_+)$}}
		
		\put (43,27) {\rotatebox{90}{\footnotesize \color{red} stable}}
		
		\put (50,27) {\rotatebox{90}{$\displaystyle \color{red}\mathcal{P}_{\rm W.}(\lambda)$}}
		\put (54,18) {\scriptsize \color{red} $\displaystyle\lambda=0$ (Bianchi type I)}
		\put (54,52) {\scriptsize \color{red} $\displaystyle\lambda\rightarrow\lambda_{\rm sup}$ (W. joins P.W.)}
		\put (72,93) { Milne}
		\put (16,17) { Kasner}
		\put (32,34) {\rotatebox{45}{\footnotesize \color{blue} unstable}}
		\put (55,57) {\rotatebox{45}{\footnotesize \color{blue} stable}}
		\put (45,0) { \large $\displaystyle \Sigma_+$}
		\put (-4,50) {\rotatebox{90}{ \large $\displaystyle A$}}

		\put (42,9) {\scriptsize $\displaystyle (2-3\gamma)/4$}
		\put (11,9) {\scriptsize $-1$}
		\put (88.5,9) {\scriptsize $0$}
		
		\put (3,45) {\rotatebox{90}{\scriptsize $\displaystyle 3(2-\gamma)/4$}}
		\put (3,18) {\rotatebox{90}{\scriptsize $0$}} 
		\put (4,93) {\rotatebox{90}{\scriptsize $1$}}
		
	\end{overpic}
	\caption{Stability relations in Bianchi type V for $\gamma\in (2/3,2)$.}
	\label{Fig:stability_BV}
\end{figure}
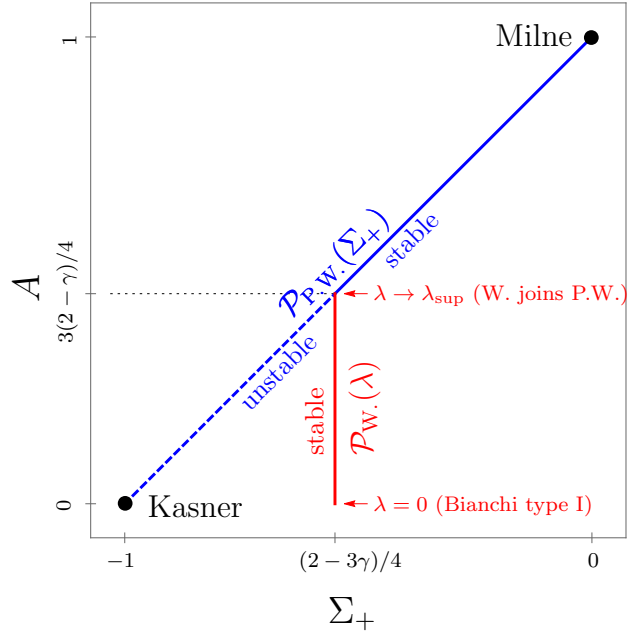

Since there is a multitude of attractors for $\gamma>2/3$, that are qualitatively very different, one may ask: 
\emph{Which initial conditions gives a Plane Wave, or a Wonderland solution, asymptotically?}

This question, that has a global character, is answered by theorem \ref{thm:BV} in the following subsection. The existence of this theorem is suggested by some geometric observations that are worth noting. First recall that $\Sigma_+ / \Theta$ is a constant of motion. Thus the shadow of each integral curve on the $\Theta\Sigma_+$-plane is a radial line segment. It is therefore worthwile to consider the projection of state space onto the $\Theta\Sigma_+$-plane. Since $\Sigma_-$ is monotone, consider without loss of generality the invariant subspace $\Sigma_-=0$. The shadow of J.E.D. is then the unit circle, whereas P.W. is the circle
\[ \Theta^2 + \left( \Sigma_++\frac{1}{2} \right)^2 = \left( \frac{1}{2} \right)^2. \]
In this projection Wonderland is a horizontal \emph{chord} of the Plane Wave circle, see figure \ref{fig:BVglobal}. Since the integral curves are radial, it is clear that one and only one attractor point is dynamically accessible for each integral curve. It follows that Wonderland is dynamically accessible only if
\begin{equation} 
\frac{\Sigma_+^2}{\Theta^2} > \frac{3\gamma-2}{3(2-\gamma)} \label{faeef} 
\end{equation}
or equivalently (in variables of the reduced system)
\[ 4\sin^2\alpha > 3\gamma-2. \]
 Similarly, it is clear that the stable part of the Plane Wave is dynamically accessible only if the condition is \emph{not} satisfied. All this is geometrically obvious by looking at figure \ref{fig:BVglobal}, where the regions that satisfy (\ref{faeef}) are shaded in green.
 	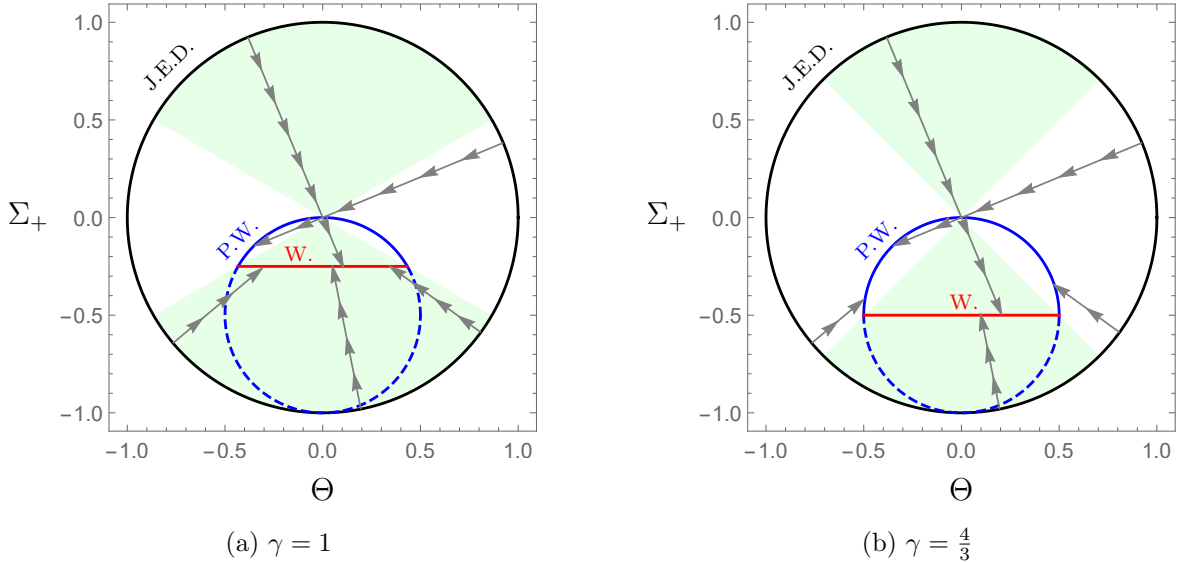
\begin{figure}[!ht]
 	\subfloat[$\gamma=1$\label{subfig-1:dummy}]{%
 		\begin{overpic}[width=0.45\textwidth]{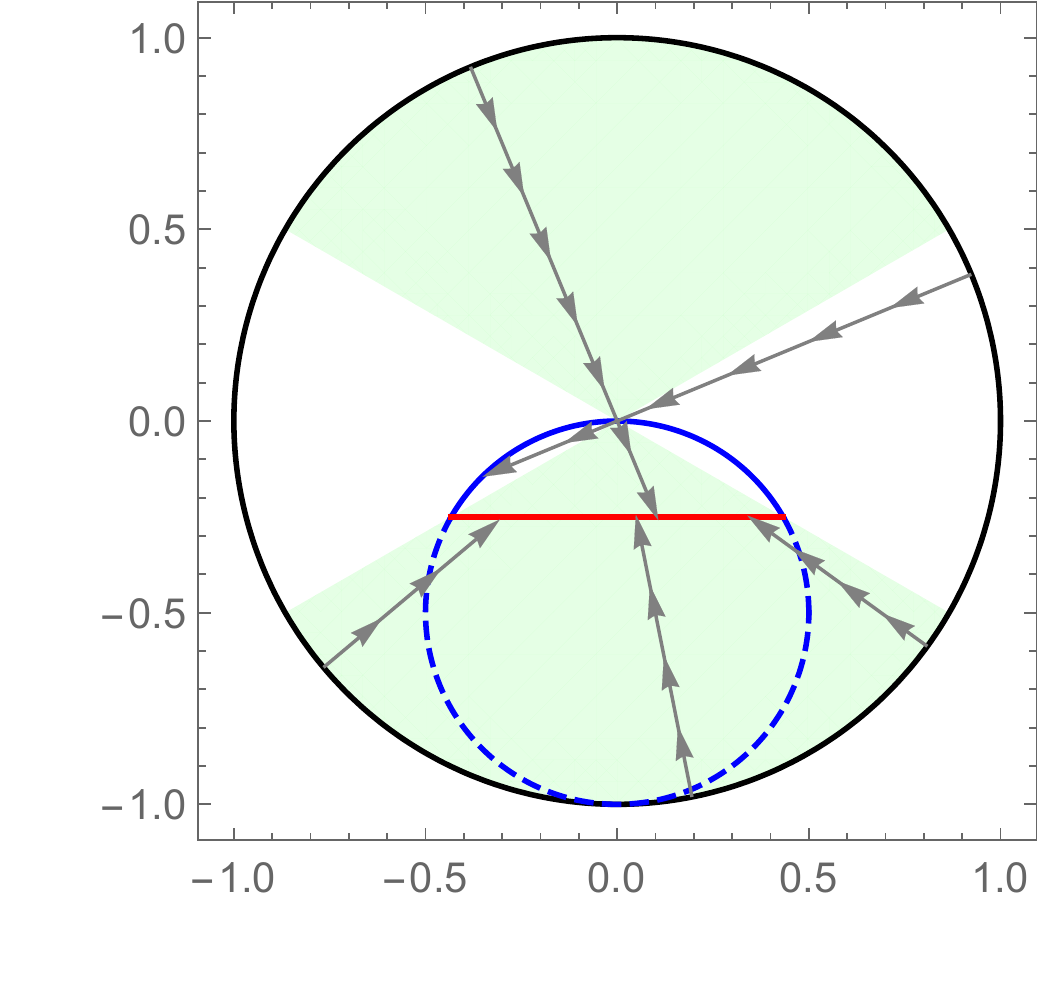}
 			\put (55,3) { $\displaystyle \Theta$}
 			\put (-2,55) { $\displaystyle \Sigma_+$}
 			\put (25,80) {\rotatebox{45}{\scriptsize \color{black} J.E.D.}}
 			\put (39,48) {\rotatebox{45}{\scriptsize \color{blue} P.W.}}			
 			\put (52,48) {\rotatebox{0}{\scriptsize \color{red} W.}}					
 		\end{overpic}
 	}\hfill
 	\subfloat[$\gamma=\frac{4}{3}$\label{subfig-2:dummy}]{%
 		\begin{overpic}[width=0.45\textwidth]{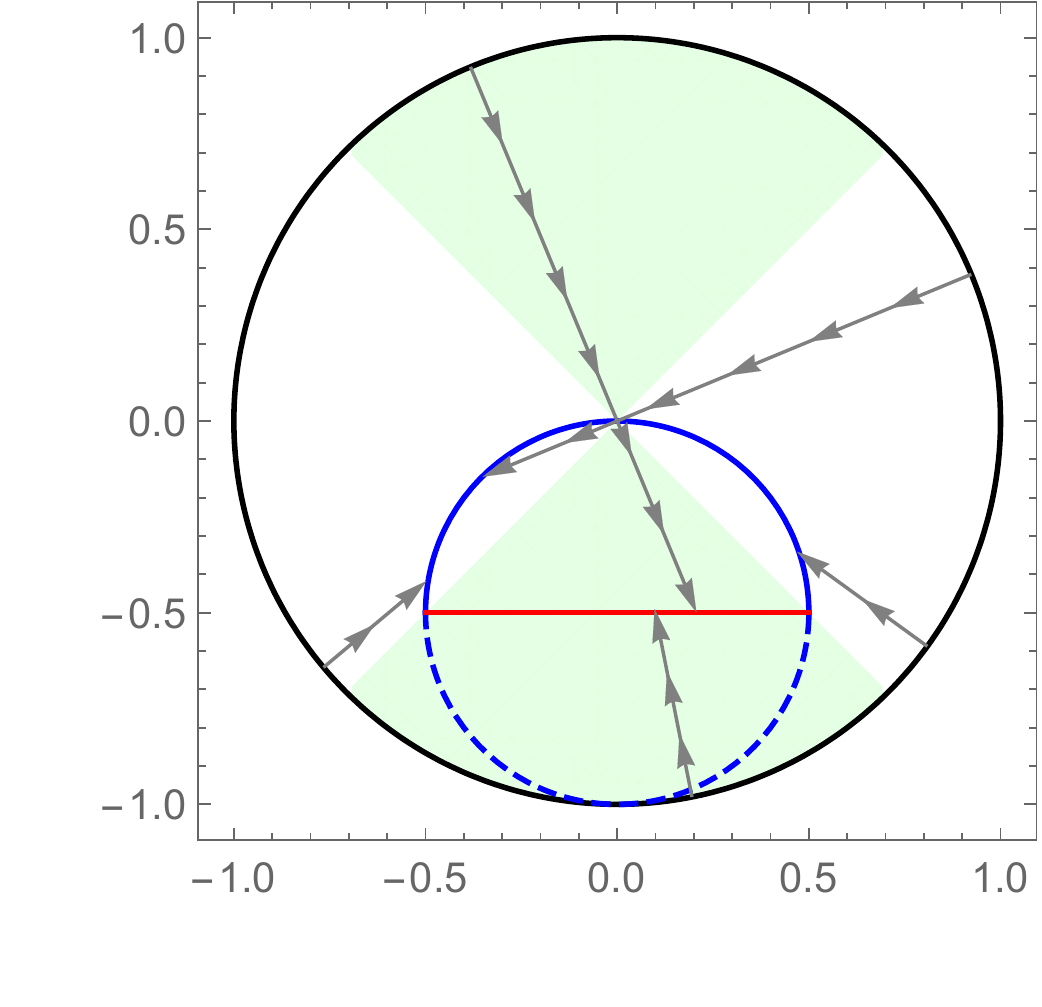}
 			\put (55,3) { $\displaystyle \Theta$}
 			\put (-2,55) { $\displaystyle \Sigma_+$}
 			\put (25,80) {\rotatebox{45}{\scriptsize \color{black} J.E.D.}}
 			\put (39,48) {\rotatebox{45}{\scriptsize \color{blue} P.W.}}			
 			\put (57,39) {\rotatebox{0}{\scriptsize \color{red} W.}}					
 		\end{overpic}
 	}
 	\caption{Projection of the invariant subspace $\Sigma_-=0$ onto the $\Theta\Sigma_+$-plane, for a dust model in (a) and radiation in (b). In the regions shaded in green the condition \ref{faeef} for dynamical access to a Wonderland attractor point is satisfied. Integral curves are radial lines (direction indicated by arrows) that goes from Jacobs' Extended Disk (J.E.D.) to an attractor point, either on Wonderland (W.) or on Plane Wave (P.W.). The stable part of P.W. is indicated by solid (blue) line, and the unstable part by a dashed (blue) line. }\label{fig:BVglobal}
 \end{figure}

Combined with the results of the local stability analysis these observations suggest that Wonderland is reached asymptotically iff \eref{faeef} is satisfied, and that a Plane Wave solution is reached asymptotically otherwise (assuming only that the solution approaches an attractor point at late times, and that $\Omega_{{\rm pf}}>0$). This is the content of theorem \ref{thm:BV} that is proved below, using monotone functions.

\subsection{Global analysis and an anisotropic hair theorem}
Also in the type V model there exists sufficient number of monotone functions to determine global behaviour of solutions, see \ref{App:Zs}. In addition to the monotone functions, the analysis rests on the observation that the quantities $(V_1/A)$ and $(\Sigma_+/\Theta)$ are equal and constant ($=\tan\alpha$). Let $(V_1)_0$ denote the value of $V_1$ at the equilibrium point. From this the above one may now infer: 
\begin{thm}[Anisotropic hairs for $2/3<\gamma<2$] 
Consider a Bianchi type V model ($A>0$) with a non-zero $j$-form field $V_1\neq 0$. Assume that $2/3<\gamma<2$, then the solutions will be asymptotically anisotropic with $V_1\rightarrow (V_1)_0\neq 0$ at late times. More specifically, if: 
\begin{itemize}
\item{} $2/3<\gamma<2$ and $\Omega_{\rm pf}=0$, then it will asymptotically approach a \emph{plane-wave} solution; 
\item{} $\Omega_{\rm pf}>0$ and $4\sin^2\alpha\leq (3\gamma-2)$, then it will asymptotically approach a \emph{plane-wave} solution;
\item{} $\Omega_{\rm pf}>0$ and $4\sin^2\alpha> (3\gamma-2)$, then it will asymptotically approach a \emph{Wonderland} solution.
\end{itemize}
\label{thm:BV}
\end{thm}
\begin{proof}
Use of monotone functions $Z_6$ and $Z_7$. 
\end{proof}

\section{Conclusion}

In this paper the evolution of $p$-form gauge fields in anisotropic space-times (Bianchi type I-VII$_h$) has been investigated. The observational evidence of some unexpected features (``anomalies'') on very large scales in the CMB, and on the other side the lack of consideration of $p$-forms in a cosmological context, motivated the investigation of the evolution of such general matter fields ($p$-forms) in anisotropic space-times. 

The general equations for a gauge field with a $j$-form field strength (with $j\in\{1,3\}$) alongside a perfect fluid obeying a $\gamma$-law equation of state and a 4-form (cosmological constant), have been computed for the first time, in an orthonormal frame. A dynamical systems approach has then been applied to the cases of Bianchi type I and V.
All self-similar cosmological solutions represented by equilibrium points have been found and their stability has been analyzed. 
In the case of Bianchi type I both a Locally Rotationally Symmetric (LRS) (``Wonderland'') and non-LRS (``the Rope'' and ``the Edge'') self-similar solutions have been found and in the physically relevant parameter region $6/5<\gamma \le 2$, a unique non-LRS solution for each value of $\gamma$ is found to be stable, in fact a global attractor. 
The Rope and the Edge possess a purely spatial field strength \emph{rotating} with respect to  the comoving inertial tetrad, and this ``vector rotation'', which is a qualitatively new feature in cosmological dynamical systems, has been identified as a cruical physical mechanism that controls the asymptotic behavior in spatially flat backgrounds with $\gamma>6/5$.

The Bianchi type V space-time, on the other hand, can only accomodate a field strength with a single spatial component. Thus there is not enough freedom for the vector to rotate and there is no Bianchi type V version of the Edge or the Rope. Wonderland, on the other hand, possesses a non-rotating vector and this solution generalizes to the Bianchi type V domain via the parameter $\lambda$, that takes the value $\lambda=0$ in the particular case of Bianchi type I. This one-parameter family of equilibrium points are attractors in the entire existence range $2/3<\gamma<2$. Furthermore, a one-parameter family of Plane Wave solutions with a null-like field strength $j$-form have been found in Bianchi type V. For $\gamma>2/3$ both Wonderland and Plane Wave type attractors are present, and the future assymptotic behavior has been determined globally using geometric observatations and monotone functions.   

As a consequence of no-hair theorems in section \ref{sec:nohair}, the considered family of minimally coupled gauge fields are not observationally relevant during an inflationary phase in the early universe, i.e. for models with $\gamma <2/3$. On the other hand, for $\gamma\in(2/3, 4/3]$ all spatially flat anisotropic attractors have a deceleration parameter $q=-1+\frac{3}{2}\gamma$ identical to the corresponding flat FLRW solution. Thus the matter and the radiation dominated epochs are interesting as potential playgrounds where a $j$-form may participate in the cosmic dynamics and produce imprints in cosmological observables. Note that a purely spatial $j$-form has equation of state $-1/3$ and thus effectively acts as spatial curvature on the background level. Thus non-trivial large-scale imprints in the CMB induced by the $j$-form via the shear tensor, analogous to those produced by spatial curvature in general Bianchi models \cite{1985MNRAS.213..917B} (which has been linked to CMB anomalies \cite{Jaffe:2005pw,Jaffe:2005gu}), are expected even in the spatially flat Bianchi type I model.

As a suggestion for further work, it is also natural to mention the specific analysis of the remaining solvable Bianchi types II,III,IV, VI$_h$ and VII$_h$. The remaining types VIII and especially IX require a slightly different approach. 

\section*{Acknowledgements}
SH was supported through the Research Council of Norway, Toppforsk
grant no. 250367: \emph{Pseudo-Riemannian Geometry and Polynomial Curvature Invariants:
Classification, Characterisation and Applications.}

%%%%%%%%%%%%%%%%%%%%%%%%%%%%%%%%%%%%%%%%%%%%%%%%%%%%%%%%%%%%%%%%%%%%%%%%%%%%
\newpage
\appendix{\huge\textbf{Appendices}}

\section{Decomposition}
The notation used in the paper is such that
\label{app:Decomp}
\begin{equation}
x_{ab}=
\left(\begin{array}{ccc}
-2x_+      & \sqrt{3}x_2&\sqrt{3}x_3\\
\sqrt{3}x_2     & x_++\sqrt{3}x_-&\sqrt{3}x_\times\\
\sqrt{3}x_3     & \sqrt{3}x_\times&x_+-\sqrt{3}x_-\\
\end{array}\right)
\end{equation}
where $x_{ab}$ is one of the traceless matrices $\pi_{ab}$ or $\sigma_{ab}$ (their normalized equivalents $\Pi_{ab}$ and $\Sigma_{ab}$ have the same structure). For the considered Bianchi type I-VII$_h$ models $n_{ab}$ can always be written on the form
\begin{equation}
\label{N}
n_{ab}=
\left(\begin{array}{ccc}
0      & 0&0\\
0     & n_++\sqrt{3}n_-&\sqrt{3}n_\times\\
0& \sqrt{3}n_\times &n_+-\sqrt{3}n_-\\
\end{array}\right)
\end{equation}

\section{Diagonal shear frame} 
\label{App:diag}
In Bianchi type I perfect fluid models the diagonal shear frame admits an inertial tetrad, i.e. $\Omega_{\rm a}=0$ for $a\,\in\{1,2,3\}$. That, however, does not work here, where the gauge field sources the shear tensor non-trivially. In fact, in order to avoid frame rotation the gauge field would need to be aligned with one of the eigen vectors of the shear tensor, which would imply restrictions on the physical degrees of freedom. But it is still possible to diagonalize the shear, without any loss of generality, by tuning the frame rotations correctly. Using \eref{B1} the full set of equations governing the Bianchi type I system (with $\Omega_\Lambda=0$) can be shown to become
\begin{eqnarray}
\label{FinalEqs1}
&\Sigma_+'=-\left(2-q\right)\Sigma_++\abs{\mathbf{V}_c}^2-2V_1^2,\\\label{FinalEqs2}
&\Sigma_-'=-\left(2-q\right)\Sigma_-+\sqrt{3}\Re\{\mathbf{V}_c^2\},\\\label{FinalEqs3}
&V_1'=-(-2\Sigma_+-q)V_1+\Im\{\mathbf{R}_c\mathbf{V}_c^*\},\\\label{FinalEqs4}
&\mathbf{V}_c'=-(\Sigma_+-q+iR_1)\mathbf{V}_c-\sqrt{3}\Sigma_-\mathbf{V}_c^*+i\mathbf{R}_cV_1,\\\label{FinalEqs5}
&\Theta'=-(2-q)\Theta,\\\label{FinalEqs6}
&\Omega_{\rm pf}'=2\left(q+1-\frac{3}{2}\gamma \right)\Omega_{\rm pf},
\end{eqnarray}
alongside the following constraints and frame rotation specifications:
\begin{eqnarray}
&1=\Sigma_{+}^2+\Sigma_-^2+\Omega_{\rm pf}+\Theta^2+V_1^2+\abs{\mathbf{V}_c}^2\label{B1Ham},\\
&\,\Theta V_1=0\label{B1C2},\\
&\,\Theta\mathbf{V}_c=0\label{B1C3},\\
&2\sqrt{3}V_1\mathbf{V}_c=i\sqrt{3}\Sigma_+\mathbf{R}_c-i\Re\{\mathbf\Sigma_\Delta\}\mathbf{R}_c^*\label{frameDef1},\\
&\Im\{\mathbf{V}_c^2\}=2\Re\{\mathbf\Sigma_\Delta\}R_1\label{frameDef2}.
\end{eqnarray}
The two last equations, which follow from the off-diagonal components of the shear propagation equation, specify the frame. 

\section{Full set of scalar equations for Bianchi 1 in G$_2$ frame}
\label{App:G2B1}
In Bianchi type I the evolution equations reduce to the following form in a frame where a G$_2$ subgroup is chosen orthogonal to $\mathbf{V}$, which is aligned along $\mathbf{e}_1$. Upon fixing the remaining gauge freedom associated with rotation about $\mathbf e_1$, as specified in the main text, equations (\ref{B1_G2:first})-(\ref{B1_G2:last}) are reproduced. By assumption, $V_1$ is non-zero and the constraint $\Theta V_1=0$ implies that the temporal component of the $j$-form is zero. 
\begin{eqnarray}
\label{B1_G2}
&\Sigma_+'=(q-2)\Sigma_+ + 3(\Sigma_2^2+\Sigma_3^2)-2V_1^2\\
&\Sigma_-'=(q-2)\Sigma_- + 2 R_1 \Sigma_\times + \sqrt{3}(\Sigma_2^2-\Sigma_3^2)\\
&\Sigma_\times'=(q-2)\Sigma_\times - 2 R_1 \Sigma_- + 2\sqrt{3} \Sigma_2 \Sigma_3\\
&\Sigma_2'=\left(q-2-3\Sigma_+-\sqrt{3}\Sigma_-\right)\Sigma_2-(\sqrt{3}\Sigma_\times-R_1)\Sigma_3 \label{heihei}\\
&\Sigma_3'=\left(q-2-3\Sigma_++\sqrt{3}\Sigma_-\right)\Sigma_3-(\sqrt{3}\Sigma_\times+R_1)\Sigma_2 \label{heihei2}\\
&V_1'=\left(q+2\Sigma_+\right)V_1\\
&\Omega_{\rm pf}'=2\left(q+1-\frac{3}{2}\gamma \right)\Omega_{\rm pf}
\end{eqnarray}
%The following equations give the Bianchi I system in a frame where a G$_2$ subgroup is chosen orthogonal to $\mathbf{V}$, which is aligned along $\mathbf{e}_1$. 
%\begin{eqnarray}
%\label{B1_G2}
%&\Sigma_+'=(q-2)\Sigma_++3\Sigma_3^2-2V_1^2\\
%&\Sigma_-'=(q-2)\Sigma_-+\sqrt{3}(2\Sigma_\times^2-\Sigma_3^2)\\
%&\Sigma_\times'=(q-2-2\sqrt{3}\Sigma_-)\Sigma_\times\\
%&\Sigma_2'=\left(q-2-3\Sigma_+-\sqrt{3}\Sigma_-\right)\Sigma_2-(\sqrt{3}\Sigma_\times-R_1)\Sigma_3\\
%&\Sigma_3'=\left(q-2-3\Sigma_++\sqrt{3}\Sigma_-\right)\Sigma_3-(\sqrt{3}\Sigma_\times+R_1)\Sigma_2\\
%&V_1'=\left(q+2\Sigma_+\right)V_1\\
%&\Omega_{\rm pf}'=2\left(q+1-\frac{3}{2}\gamma \right)\Omega_{\rm pf}
%\end{eqnarray}

\section{Eigenvalues}
\label{App:EigVal}

\begin{table}[H]
	\centering
	\resizebox{\textwidth}{!}{\begin{tabular}{lll}
			\toprule
			\multicolumn{3}{c}{\textbf{Eigenvalues of equilibrium sets in Bianchi type I with pure temporal field} ($V_1=\mathbf{V}_c=0$)} \\
			\hline
			Name (abbr.) &$\gamma$ &eigenvalues $\{l_1,l_2,l_3\}$\\
			\hline
			flat FLRW &$[0,2)$&$\{\,\frac{3}{2}(\gamma-2),\,\frac{3}{2}(\gamma-2),\,\frac{3}{2}(\gamma-2)\}$ \Tstrut \\
			Kasner (K.) &$[0,2)$&$\{\,3(2-\gamma),0,0\}$\\
			Jacobs' Ext. D. (J.E.D.) &$[0,2)$&$\{\,3(2-\gamma),0,0\}$
         % \\ Jacobs' Sphere (J.S.) &$2$&$\{0,0,0\}$\\
	\end{tabular}}
	\caption{Table of eigenvalues of equilibrium points in Bianchi type I with timelike field ($V=\mathbf{V}_c=0$). The independent variables used in the linearization are $\{ \Sigma_+, \Sigma_-, \Theta\}$.}
	\label{tab:EigVal_BI_temp}
\end{table}

\begin{table}[H]
	\centering
	\resizebox{\textwidth}{!}{\begin{tabular}{lll}
			\toprule
			\multicolumn{3}{c}{\textbf{Eigenvalues of equilibrium sets in Bianchi type I with pure spatial field} ($\Theta=0$)} \\
			\hline
			Name (abbr.) &$\gamma$ &eigenvalues $\{l_1,l_2,l_3,l_4,l_5\}$\\
			\hline
			Wonderland (W.)&$(\frac{2}{3},2)$&$\{\frac{3}{4}\left(\gamma-2-\Gamma(\gamma,0)\right),\,\frac{3}{2}(\gamma-2),\,\frac{3}{2}(\gamma-2),\,\frac{15}{4}(\gamma-\frac{6}{5}),\,\frac{3}{4}\left(\gamma-2+\Gamma(\gamma,0)\right)\}$ \Tstrut \\
			The Rope (R.)&$(\frac{6}{5},\frac{4}{3})$&$\{9(\gamma-\frac{4}{3}),\,\frac{3}{4}\left(-2+\gamma+\sqrt{A+2\sqrt{B}}\right),\,\frac{3}{4}\left(-2+\gamma+\sqrt{A-2\sqrt{B}}\right),\,\frac{3}{4}\left(-2+\gamma-\sqrt{A+2\sqrt{B}}\right),\,\frac{3}{4}\left(-2+\gamma-\sqrt{A-2\sqrt{B}}\right)\}$\\
			The Edge (E.)&$[0,2]$& $\{\frac{1}{2} \left(-1+i \sqrt{23}\right),\,\frac{1}{2} \left(-1- i \sqrt{23}\right),\,-1,\,0,\,3(\frac{4}{3}-\gamma\}$\\
			Flat FLRW &$[0,2)$&$\{\frac{3}{2}(\gamma-2),\,\frac{3}{2}(\gamma-2),\,\frac{3}{2}(\gamma-2),\,\frac{3}{2}(\gamma-2),\,\frac{3}{2}(\gamma-\frac{2}{3})\}$\\
			Kasner (K.) &$[0,2)$&$\{0,\,3(2-\gamma),\,-2 \sqrt{3} \Sigma_-,\,\sqrt{3} \Sigma_- -3 \Sigma_+,\, 2(1+\Sigma_+)\}$ \\
			%Jacobs' disk (J.D.) & $2$ & $\{0,\,0,\,-2 \sqrt{3} \Sigma_-,\,\sqrt{3} \Sigma_- -3 \Sigma_+,\, 2(1+\Sigma_+)\}$\\ \hline 
	\end{tabular}}
	\caption{Table of eigenvalues of equilibrium points in Bianchi type I with spacelike field ($\Theta=0$). The independent variables used in the linearization are $\{ \Sigma_+, \Sigma_-,  \Sigma_\times,\Sigma_3,V_1 \}$}
	\label{tab:EigVal_BI_spatial}
\end{table}
Above $A(\gamma)$ and $B(\gamma)$ are defined such that
\begin{eqnarray}
&A(\gamma)=(2-\gamma) (\gamma (18 \gamma-97)+90)\\
&B(\gamma)=(\gamma-2)^2 (3\gamma ( \gamma-\frac{4}{3}) (27 \gamma (\gamma+4)-136)+16)
\end{eqnarray}

\begin{table}[H]
	\centering
	\resizebox{\textwidth}{!}{\begin{tabular}{clclr}
			\toprule
			\multicolumn{4}{c}{\textbf{Eigenvalues of equilibrium sets in Bianchi type V}} \\
			\hline
			Type&Name (abbr.) &$\gamma$ &eigenvalues $\{l_1,l_2,l_3,l_4\}$ \\
			\hline
			V&Plane wave (P.W.) &$[0,2]$&$\{-2(1+\Sigma_+),-2(1+\Sigma_+),0,2-3\gamma-4\Sigma_+\}$\\
			V&Wonderland (W.) &$(\frac{2}{3},2)$&$\{\frac{3}{4}\left(\gamma-2-\Gamma(\gamma,\lambda)\right),\frac{3}{4}\left(\gamma-2+\Gamma(\gamma,\lambda)\right),0,\frac{3}{2}(\gamma-2)\}$\\
			V&open FLRW &$\frac{2}{3}$&$\{-2,-2,0,0\}$\\
			I&flat FLRW&$[0,2)$&$\left\{\frac{3}{2}(\gamma-2),\frac{3}{2}(\gamma-2),\frac{3}{2}\left(\gamma-\frac{2}{3}\right),\frac{3}{2}\left(\gamma-\frac{2}{3}\right)\right\}$\\
			I&Jacobs' Ext. D. (J.E.D.)&$[0,2)$&$\left\{3(2-\gamma), 2 (1+\Sigma_+ ), 0, 0 \right\}$\\
			I&Jacobs' Sphere (J.S.) &$2$&$\{2(1+\Sigma_+), 0,0,0 \}$\\
	\end{tabular}}
	\caption{Table of eigenvalues for equilibrium points in Bianchi type V. The independent variables used in the linearization are $\{A, \Sigma_+, \Sigma_-,  V_1 \}$. For J.E.D. and J.S. the reduced system \eref{BVRed}-\eref{BVRedlast}, with $\{\Sigma_-, \nu, \eta, \alpha\}$ as independent variables, has been used. }
	\label{tab:EigVal:V}
\end{table}

In the above 
\begin{equation}
\Gamma (\gamma,\lambda)=\sqrt{6(2-\gamma)\left((\gamma-2)(\gamma-\frac{5}{6})+\frac{16}{3}(\gamma-\frac{2}{3})\lambda^2\right)}
\end{equation}

\section{Monotone functions}\label{App:Zs}
The dynamical systems of Bianchi type I and V possess several monotone functions which can be found by using the same techniques as in \cite{Hervik:2010uh}. The exact solutions from which they are constructed are given, in addition to the range of $\gamma$ which the functions are monotone. 
\subsection{Monotone functions for Bianchi type I} 
In the following, the function $\phi$ is constructed the following way: $\phi=1-(\Sigma_+)_0\Sigma_+-(\Sigma_-)_0\Sigma_-$, where $(\Sigma_+)_0$ and $(\Sigma_-)_0$ mean the constant values corresponding to the equilibrium point. Now note that, using the scalar product:
\[ \phi=1-((\Sigma_+)_0,(\Sigma_-)_0)\cdot(\Sigma_{+},\Sigma_{-})\geq 1-\sqrt{(\Sigma_+)_0^2+(\Sigma_+)_0^2}\sqrt{\Sigma_{+}^2+\Sigma_{-}^2}\geq 0,\]
where equality only holds for both $((\Sigma_+)_0,(\Sigma_-)_0)$ and  $(\Sigma_{+},\Sigma_{-})$ on the Kasner circle. Since none of the equilibrium points below are on the Kasner circle, $\phi\geq \phi_{{\rm Min}}>0$.

\paragraph{General, $2/3\leq \gamma\leq 2$:}
\begin{eqnarray}
Z_1&=&\frac{\Sigma_\times\Sigma_3^2V_1^3}{\Omega^3_{\rm pf}}, \qquad Z_1'=3(3\gamma-4)Z_1, 
\label{z1t} 
\eeq
This function is increasing for $\gamma>4/3$ and decreasing for $\gamma<4/3$. 
\paragraph{Wonderland, $2/3\leq \gamma\leq 6/5$:}
\beq
Z_2&=&\frac{V_1^{2m}\Omega_{\rm pf}}{(1+m\Sigma_+)^{2(1+m)}}, \qquad m=\frac 14(3\gamma-2)\\
\frac{Z_2'}{Z_2}&=&\frac{1}{1+m\Sigma_+}\left\{4(\Sigma_++m)^2+\frac{3(3\gamma+2)}{8}\left[2(2-\gamma)(\Sigma_-^2+\Sigma_{\times}^2)+(6-5\gamma)\Sigma_3^2\right]\right\} \nonumber 
\eeq
\paragraph{The Rope, $6/5\leq \gamma\leq 4/3$:} 
\beq
Z_3 &=&\frac{\Sigma_3^{2(5\gamma-6)}V_1^{2(9\gamma-10)}\Omega_{\rm pf}^b}{\phi^a}, 
\eeq
where
\beq
\phi& =&1+\frac{(3\gamma-2)}{4}\Sigma_++\frac{\sqrt{3}(5\gamma-6)}{4}\Sigma_-, \\ a&=&\frac{4(16\gamma-7\gamma^2-8)}{(2-\gamma)}, ~~b=\frac{4(4-3\gamma)}{(2-\gamma)}. \\
\frac{Z_3'}{Z_3}&=&\phi^{-1}\Bigg\{8\left[\frac{(3\gamma-2)}{4}X+\frac{\sqrt{3}(5\gamma-6)}{4}Y\right]^2
\\ && +6(16\gamma-7\gamma^2-8)(X^2+Y^2) +3a(4-3\gamma)\Sigma_{\times}^2\Bigg\},
\\ &&X=\Sigma_++\frac{(3\gamma-2)}{4}, \qquad Y=\Sigma_-+\frac{\sqrt{3}(5\gamma-6)}{4}.
\eeq
\paragraph{The Edge, $4/3\leq \gamma\leq 2$:} 
\beq
Z_4&=&\frac{\Sigma_3^2V_1^6}{\phi^8}, \qquad \phi=1+\frac 12\Sigma_++\frac{1}{2\sqrt{3}}\Sigma_-, \\
\frac{Z_4'}{Z_4}&=&\phi^{-1}\left[8(X^2+Y^2)+(\sqrt{3}X+Y)^2+4(3\gamma-4)\Omega_{\rm pf}\right], \\
&& \quad X=\Sigma_++\frac 12, \quad Y=\Sigma_-+\frac{1}{2\sqrt{3}}. 
\end{eqnarray}
\subsection{Monotone functions for Bianchi type V}
This section refers to the reduced system, and all functions are monotone for a fixed value of $\alpha$. 
 \paragraph{General, $2/3\leq \gamma\leq 2$:}
\beq
Z_5=\Sigma_-, \qquad Z_5'=-(2-q)Z_5.
\eeq
\paragraph{Wonderland, $\sin^2\alpha\geq\frac14(3\gamma-2)\geq 0$:}
\beq
Z_6&=& \frac{\eta^b\Omega_{\rm pf}^c}{(1+n\nu)^a}, \quad n=\frac{3\gamma-2}{4\sin^2\alpha}\\
&& a=\frac 18\left[{16\sin^2\alpha-(3\gamma-2)^2}\right], \quad b=\frac 38(3\gamma-2)(2-\gamma),\nonumber \\ 
&& c=\sin^2\alpha-\frac14(3\gamma-2), \nonumber \\
\frac{Z_6'}{Z_6}&=&\frac{2-\gamma}{1+n\nu}\left[3(\nu+n)^2+\frac32a\Sigma_-^2\right] 
\eeq
\paragraph{Plane waves, $\sin^2\alpha\leq\frac14(3\gamma-2)$:}
\beq
Z_7&=& \frac{\eta^2}{(1+\nu\sin\alpha)^2}, \\
\frac{Z_7'}{Z_7}&=&\frac{4}{1+\nu\sin\alpha}\left[ (\nu+\sin\alpha)^2+\cos^2\alpha\Sigma_-^2+\left(\frac{3\gamma-2}{4}-\sin^2\alpha\right)\Omega_{\rm pf}\right]
\eeq

%\section{Conventions}
%\begin{itemize}
%\item self-similar
%\item space-time
%\item energy-momentum
%\item no-hair
%\item the Einstein Field Equations (not Einstein's egs. / Einstein Eqs.)
%\item co-moving
%\item abbreviations: W.,P.W., J.E.D., J.S., etc (period after each letter).
%\end{itemize}

\end{document}